\documentclass[]{llncs}
\usepackage[utf8]{inputenc}      
\usepackage[T1]{fontenc}         
\usepackage{microtype}           
\usepackage[english]{babel}      
\usepackage{graphicx}            
\usepackage{cite}                
\usepackage{hyperref}            
\usepackage{comment}             

\title{An Alternating Automaton for First-Order Linear Temporal Logic}
\subtitle{Tech Report}

\author{Yannick Lebrun \and Raphaël Khoury \and Sylvain Hallé%
}
\institute{%
Laboratoire d'informatique formelle \\
Université du Québec à Chicoutimi, Canada%
}

\usepackage{amsmath,amsfonts,amssymb}
\usepackage{subfig}
\usepackage{rotating}

\usepackage[retainorgcmds]{IEEEtrantools}

\usepackage[x11names]{xcolor}
\usepackage{todonotes}

\usepackage[normalem]{ulem}  

\newtheorem{dfn}{Definition}
\newtheorem{lem}{Lemma}
\newtheorem{prs}{Proposition}
\newtheorem{thm}{Theorem}
\newtheorem{runtime-property}{Runtime Property}
\newtheorem{rmk}[thm]{Remark}

\newcommand{\tab}{\quad}

\newcommand{\LTLE}{LTL-FO$^+$}

\newcommand{\G}{\mbox{\bf G\,}}
\newcommand{\X}{\mbox{\bf X\,}}
\newcommand{\F}{\mbox{\bf F\,}}

\newcommand{\U}{\mbox{\bf \,U\,}}

\newcommand{\V}{\mbox{\bf \,V\,}}

\newcommand{\Farg}[1]{\F \, #1}
\newcommand{\Garg}[1]{\G \, #1}
\newcommand{\Xarg}[1]{\X \, #1}

\newcommand{\R}{\mathbf{R}}
\newcommand{\Rarg}[2]{#1 \, \R \, #2}
\newcommand{\Uarg}[2]{#1 \, \U \, #2}
\newcommand{\Earg}[3]{\exists_{#1} #2 : #3}
\newcommand{\Aarg}[3]{\forall_{#1} #2 : #3}
\newcommand{\TRUE} {\mathbf{true}}
\newcommand{\FALSE}{\mathbf{false}}

\newcommand{\subformulaset}[1]{sub(#1)}
\newcommand{\accstate}{\top}
\newcommand{\rejstate}{\bot}

\newcommand{\pfuncset}[2]{\mathcal{PF}_{#1 \rightarrow #2}}
\newcommand{\transfunc}[3]{\rho\big((#1, #2), #3\big)}
\newcommand{\transfuncDeux}[2]{\rho\big(#1, #2\big)}

\newcommand{\timedepth}[1]{depth(#1)}
\newcommand{\logicprop}[3]{P[#1, #2](#3)}
\newcommand{\rhorun}[3]{r\big((#1, #2), #3\big)}

\newcommand{\image}[1]{I}

\newcommand{\showboxcond}{$\Box$ }

\graphicspath{{fig/}}

\begin{document}

\maketitle
\begin{abstract}
In this paper we give automata-based representation of LTL-FO$^+$ properties. LTL-FO$^+$ is an extension of LTL that includes first-order quantification over bounded variable, thus greatly increasing the expressivity of the language.  An automata representation of this formalism allows greater ease in writing and understanding  properties, as well as in performing manipulations, such as negation or emptiness checking. The automata representation of an  LTL-FO$^+$ formula has finite size regardless of the domain of quantified variables, and the number of states that is linear in the size of the property.\end{abstract}

\section{Introduction} 
\LTLE{} \cite{DBLP:journals/tsc/HalleV12} is a  formal language used for the specification of trace properties which distinguishes itself from other representations by its exceptional expressiveness.  It  allows users to state finer relationships between the different elements of several messages in complex event trace. For example, in \cite{DBLP:journals/tsc/HalleV12} \LTLE{} is used to express properties related to  XML messages traces generated by web services. Such properties cannot be stated with a less expressive formalism such as LTL.

However, the on-the-fly verification algorithm is based on the decomposition and rewriting of formul\ae{}. It is highly intensive in its space consumption, with multiple manipulation being performed even when processing messages that have no bearing on the validity of the formula.
Indeed, for some formul\ae{}, the evaluation tree expands indefinitely, and equivalent subtrees, which could be pruned, are hard to identify.  Furthermore, the  elaborate syntax of \LTLE{} can make it difficult to state and read properties (see for example \cite{IsolaCorfu}).

An automat-based representation of \LTLE{} formul\ae{} would thus ease both the writing and reading of formul\ae{}. The space and time overhead of the verification process would also be optimized, since the validity of a formula could be ascertained by maintaining a list of current valuations of and current state or states. Furthermore, an automata representation will allow multiple useful manipulations to be performed with ease, notably counterexamples generation, emptiness checking, intersection and negation of properties.

The automata we propose is a variation of Vardi's alternating automata \cite{Vardi:alternating}, which we have enriched with first-order quantifiers over a finite set of formula variables. This makes it easier to express intricate formul\ae{} over complex events, where each event consists of a XML object with possibly multiple valuations for each path.  Like the alternating automata, the proposed automata distinguishes between \textit{existential} and \textit{universal} transitions. Existential transitions are analogous to non-deterministic transitions in regular Büchi automata. Upon encountering such a transition, the automata can be though of as choosing between multiple destination states. Conversely, when encountering a universal transition, the automata continues its run in both target states simultaneously. Because of the presence of universal and existential quantifiers, a run over an alternating automata generates a tree of states. A run is accepting of there \textit{exists} at least one tree for which \textit{every} branch visits an accepting state infinitely often. While alternating automata are equally expressive as non-deterministic Büchi automata, we show in this paper how the notions of existential and universal transition can be used to model the quantifiers present in \LTLE{} formul\ae{}. The automata is additionally enriched with a partial function mapping formula variables to variables to their values. This function is manipulated by the automata's transition as the input sequence is read and consulted to determine the truth value of elementary propositions.

In this paper, we show how to extend Vardi's alternating automata to accommodate  the greater expressivity of \LTLE{}. Section \ref{sec:ltlfe} provides background information about \LTLE{}. Section \ref{sec:altrep} surveys existing automata representations for other formal logics. In Section \ref{sec:construction}, we show how to construct an modified alternating automata from an \LTLE{} formula, such that the automata accepts exactly the same set of input sequences as the original property. Section \ref{sec:proof} sketches out a proof of correctness. Concluding remarks are given in Section  \ref{sect:conclu}.

\section{The First-Order Temporal Logic \LTLE{}} 
\label{sec:ltlfe}
\LTLE{}, a first-order extension of a well-known logic called Linear Temporal Logic (LTL); LTL has already been suggested for the static verification of web service interface contracts \cite{Nakajima05,DBLP:conf/www/FuBS04,DBLP:journals/re/Robinson06}.

LTL has been introduced to express properties about sequences of states in systems called Kripke structures \cite{clarke-mc}.  In the current case, the states under consideration are XML objects termed messages. Let us denote by $M$ the set of XML messages. A sequence of messages  $m_1, m_2 \dots$, where $m_i \in M$ for every $i \geq 1$, is called a message trace. We write $m_i$ to denote the $i$-th message of the trace $\overline{m}$, and $\overline{m}^i$ to denote the trace obtained from $\overline{m}$ by starting at the $i$-th message.

A domain function is used to fetch and compare values inside a message; it receives an argument $\pi$ representing a {\em path} from the root to some element of the message.  This path is defined using standard, XPath 1.0 notation. Formally, if we let $\image{}$ be a domain of values, and $\Pi$ be the set of XPath expressions, the domain function $Dom$ is an application $M \times \Pi \rightarrow 2^{\image{}}$ which, given a message $m \in M$ and a path $\pi \in \Pi$, returns a subset $Dom_m(\pi)$ of $\image{}$, representing the set of values appearing in message $m$ at the end of the path $\pi$. For example, if we let $\Pi$ be the set of XPath formul\ae{}, $\pi \in \Pi$ be the particular formula ``/message/stock/name'', and $m \in M$ be the following message:

\vskip 10pt
\begin{center}
\begin{tabular}{l}
\small
$<$message$>$\\
\tab $<$action$>$placeBuyOrder$<$/action$>$\\
\tab $<$stock$>$\\
\tab \tab $<$name$>$stock-1$<$/name$>$\\
\tab \tab $<$amount$>$123$<$/amount$>$\\
\tab $<$/stock$>$\\
\tab $<$stock$>$\\
\tab \tab $<$name$>$stock-2$<$/name$>$\\
\tab \tab $<$amount$>$456$<$/amount$>$\\
\tab $<$/stock$>$\\
$<$/message$>$
\end{tabular}
\end{center}
\vskip 10pt

\noindent then $Dom_m(\pi) = \{\mbox{stock-1}, \mbox{stock-2}\}$.

\LTLE{}'s syntax is based on classical propositional logic, using the connectives $\neg$ (``not''), $\vee$ (``or''), $\wedge$ (``and''), $\rightarrow$ (``implies''), to which four temporal operators have been added. An \LTLE{} formula is a well-formed combination of these operators and connectives, according to the usual construction rules:
\begin{dfn}[Syntax]
\begin{enumerate}
\item If $x$ and $y$ are variables or constants, then
  $x=y$ is a \LTLE{} formula;
\item If $\varphi$ and $\psi$ are \LTLE{} formul\ae{}, then $\neg \varphi$,
  $\varphi \wedge \psi$, $\varphi \vee \psi$, $\varphi \rightarrow \psi$,
  $\G \varphi$,
  $\F \varphi$, $\X \varphi$,
  $\varphi \U \psi$, $\varphi \V \psi$
  are \LTLE{} formul\ae{};
\item If $\varphi$ is a \LTLE{} formula, $x_i$ is a free variable in $\varphi$,
  $p \in \Pi$ is a XPath  value,
  then $\exists_p x_i: \varphi$ and $\forall_p x_i: \varphi$ are \LTLE{} formul\ae{}.
\end{enumerate}
\end{dfn}

The semantics of an \LTLE{} formula is given with respect to a partial function $p : V \rightarrow \image{}$ that assigns every free variable in the formula. Let's denote by $\logicprop{p}{\psi}{\overline{m}}$ the predicate that outputs $\TRUE$ if the message trace $\overline{m}$ satisfies $\psi$ given $p$ and $\FALSE$ otherwise. The semantics of \LTLE{} is then given as:
\begin{table}
\begin{eqnarray*}
\logicprop{p}{c_1 = c_2}{\overline{m}}
&\Leftrightarrow&
p(c_1) = p(c_2)\\
\logicprop{p}{c_1 \neq c_2}{\overline{m}}
&\Leftrightarrow&
p(c_1) \neq p(c_2)\\
\logicprop{p}{\mu \lor \eta}{\overline{m}}
&\Leftrightarrow&
\logicprop{p}{\mu}{\overline{m}}
\lor
\logicprop{p}{\eta}{\overline{m}};\\
 \logicprop{p}{\mu \land \eta}{\overline{m}}
&\Leftrightarrow&
\logicprop{p}{\mu}{\overline{m}}
\land
\logicprop{p}{\eta}{\overline{m}};\\
 \logicprop{p}{\Xarg{\psi}}{\overline{m}}
&\Leftrightarrow&
\logicprop{p}{\psi}{\overline{m}^2};\\
\logicprop{p}{\Uarg{\mu}{\eta}}{\overline{m}}
&\Leftrightarrow&
\logicprop{p}{\eta}{\overline{m}} \lor
\Big(
\logicprop{p}{\mu}{\overline{m}} \land
\logicprop{p}{\Uarg{\mu}{\eta}}{\overline{m}^2}
\Big);\\
\logicprop{p}{\Rarg{\mu}{\eta}}{\overline{m}}
&\Leftrightarrow&
\logicprop{p}{\mu \land \eta}{\overline{m}} \lor
\Big(
\logicprop{p}{\eta}{\overline{m}} \land
\logicprop{p}{\Rarg{\mu}{\eta}}{\overline{m}^2}
\Big);\\
\logicprop{p}{\Earg{\pi}{x}{\psi}}{\overline{m}}
&\Leftrightarrow&
\bigvee_{x_i \in Dom_m(\pi)}
\logicprop{p \cup \lbrace(x, x_i)\rbrace}{\psi}{\overline{m}};\\
\logicprop{p}{\Aarg{\pi}{x}{\psi}}{\overline{m}}
&\Leftrightarrow&
\bigwedge_{x_i \in Dom_m(\pi)}
\logicprop{p \cup \lbrace(x, x_i)\rbrace}{\psi}{\overline{m}}.
\end{eqnarray*}
\caption{Semantics of \LTLE{}}
\label{tab:semantics}
\end{table}

When $p$ is clear from context we write $\overline{m} \models \varphi$ to indicate that the trace $\overline{m}$ satisfies $\varphi$. As usual, we define the semantics of the other connectors with the following identities:
$\varphi \wedge \psi \equiv \neg(\neg \varphi \vee \neg \psi)$, 
$\varphi \rightarrow \psi \equiv \neg \varphi \vee \psi$, 
$\G \varphi \equiv \neg (\F \neg \varphi)$, 
$\varphi \V \psi \equiv \neg (\neg \varphi \U \neg \psi)$, 
$\forall_p x : \varphi \equiv \neg (\exists_p x : \neg \varphi)$.



Boolean connectives carry their usual meaning. The temporal operator {\bf G} means ``globally'': the formula $\mbox{\bf G}\,\varphi$ means  $\varphi$ holds for every message of the trace. The operator {\bf F} means ``eventually''; the formula $\mbox{\bf F}\,\varphi$ holds if $\varphi$ holds for some future message of the trace. The operator {\bf X} means ``next'' and \X{} $\varphi$ holds whenever $\varphi$ holds in the next message of the trace. Finally, the {\bf U} operator means ``until'' and formula $\varphi\,\mbox{\bf U}\,\psi$ holds if $\varphi$ holds for every messages until some message satisfies $\psi$.

Any \LTLE{} formula also has an equivalent \emph{negation normal form}. An \LTLE{} formula is in \emph{negation normal form} if it doesn't contain the operators $\F$, $\G$ and $\rightarrow$, and if all negations $\neg$ are pushed inside until they precede equalities.
 ~~~~~~~~~ Definition : Negation normal form ~~~~~~~~ %
\begin{dfn}
An \LTLE{} formula is in \emph{negation normal form} if it doesn't contain the operators $\F$, $\G$ and $\rightarrow$, and if all negations $\neg$ are pushed inside until they precede equalities.
\end{dfn}
 ~~~~~~~~~~~~~~~~~~~~~~~~~~~~~~~~~~~~~~~~~~~~~~~~~~~~ %

We identify $\neg(x = y)$ with $x \neq y$ in order to eliminate the operator $\neg$ completely. This form is always obtainable with the use the identities $\Farg{\psi} \equiv \Uarg{\TRUE}{\psi}$, $\Garg{\psi} \equiv \Rarg{\FALSE}{\psi}$ and $\psi_1 \rightarrow \psi_2 \equiv \neg\psi_1 \lor \psi_2$, and by remembering that $\lor$, $\U$, and $\exists$ are the dual of $\land$, $\R$, and $\forall$ respectively. It what follows, we consider only properties in negation normal form.

The notion of \emph{temporal depth} of formul\ae{} serves as the basis for induction in a number of proofs.
\begin{dfn}
The \emph{temporal depth} of an \LTLE{} formula $\varphi$, denoted $\timedepth{\varphi}$, is the maximal number of nodes associated with a temporal operator ($\F$, $\G$, $\X$, $\U$ or $\R$) that can be observed in a branch of $\varphi$.
\end{dfn}

Variable assignments can be represented mathematically as a partial function from a set of variables to a set of values. We denote by $V$ the set of variable that occur in a \LTLE{} formula and by $\image{}$ the domain of values that may appear in a message. We write $\pfuncset{V}{\image{}}$, for the set of partial functions from $V$ to $\image{}$. The state of all variables in $V$ during a run of the automaton $A_{\varphi}$ can always be represented by a partial function $p : V \rightarrow \image{}$. Some variables may not be assigned yet, but $p$ is updated continuously as $A_{\varphi}$ reads the input trace, with $p(x)$ representing the valuation of variable  $x\in V$. Abusing the notation we write $p(v)=v$ for any constant value $v\in \image{}$ that occurs in a formula.

\subsection{Monitoring \LTLE{}}

By repeatedly applying the classical semantic rules of LTL, the evaluation of an LTL formula $\varphi$ on a trace $\sigma$ induces a tree. For example, in the case of the formula $\G (a \rightarrow \X b)$ evaluated on the trace cab, the top-level operator of that formula, \textbf{G}, corresponds to the top-level node of the tree. According to the semantics of LTL, $\G \varphi$ is true if and only if $\varphi$ is true for every suffix of the current trace. The tree hence spawns three child nodes, corresponding to the evaluation of $a \rightarrow X b$ for traces \textsl{cab}, \textsl{ab} and \textsl{b}, respectively. Taking the first such child node, the the top-level operator now becomes $\rightarrow$; this operator evaluates to when, on the current trace, either a evaluates to $\bot$ or $\X b$ evaluates to $\top$. This, in turn, spawns two child nodes corresponding to each condition, and so on. Ultimately, only equalities on values remain, and the trough value of each subformula can then be obtained by combining and propagating values towards the top of the tree.



\section{Related Representations} \label{sec:altrep} 
Büchi Automata \cite{buchi}, extend finite non-deterministic automata to infinite traces and  provide an automata representation of LTL formul\ae{}. Multiple translation algorithms exist \cite{Vardi:alternating}, and the automaton's  state set is linear in the size of the property under consideration. Unlike automata that recognize finite sequences, the nondeterminism of the automata is essential to its expressive power. An infinite sequence is valid if it enters an accepting state infinitely often on at least one of its possible runs.

Vardi \cite{Vardi:alternating},  further suggested extending the Büchi Automata to include both existential and universal transitions. When encountering an existential choice for a given input token, the automata non-deterministically chooses one of them, and can then be though of as being in either one of its multiple possible destination states (thus behaving in the same manner as a non-deterministic Büchi automata).  When encountering a universal transition, the automata's run continues simultaneously in each one of the universal transition' s destination state. A run of an alternating automata generates a tree, and an infinite input sequence is accepting of the \textit{exists} at least one tree for which \textit{every} branch visits an accepting state infinitely often. While alternating automata are no more expressive than regular Büchi automata, they can be exponentially more concise.

A more expressive representation was proposed by Barringer et al. \cite{qea}. The devised the Quantified Event Automata, an automata enriched with quantified variables.  Each variable is associated with a domain of values that have been observed so far, used to determine acceptance. Compared with \LTLE{}, QEA are strictly less expressive because QEA restrict the position on quantifiers in the formula.

Cassar et al. \cite{larva} introduce  dynamic automata with timers and events (DATEs) which serves as basis for verification in the LARVA system. Using Dates, specific events can engender a duplication of the property automata under consideration. For example, in a scenario in which multiple users interact with a service using a specific protocol, modeled by an automata. A new automaton instance would be  generated each time a user logs onto the system and initiates the communication protocol. The expressivity of \LTLE{} is orthogonal to that of DATEs. The quantifiers of present in \LTLE{} formul\ae{} allows it to process traces in which a single event contains multiple instances with the same name, which is not possible with DATEs. However, DATEs posses clocks and internal variables that allow them to  verify some behaviors that cannot be stated with \LTLE{}.

\section{Automata-Theoretic Representation of \LTLE{}}\label{sec:construction}

In this section, we show how, given an  \LTLE{} formula $\varphi$, one can build a \emph{modified alternating B\"uchi automaton} $A_\varphi = (\Sigma, V, \image{}, S, s^0, \rho, F)$ such that the language recognized by $A_\varphi$ is exactly the set of message traces satisfying $\varphi$.
The alphabet $\Sigma$ is the set $M$ of all XML messages. The set $V$ consists of all the variables that compose $\varphi$. The set $C$ includes all the constants present $\varphi$ as well as any value that can be assigned to the variables in $V$ when the automaton reads a message $m \in \Sigma$. The set $S$ of states consists of all subformulas of $\varphi$, denoted by $\subformulaset{\varphi}$, defined recursively as follows:\\
\begin{tabular}{lcccc}
 $\varphi$ &$\in$& $\subformulaset{\varphi}$; && \\
 $\mu \lor \eta$ &$\in$&  $\subformulaset{\varphi}$
&$\Rightarrow$&
$\mu, \eta \in \subformulaset{\varphi}$;\\
 $\mu \land \eta$ &$\in$&  $\subformulaset{\varphi}$
&$\Rightarrow$&
$\mu, \eta \in \subformulaset{\varphi}$;\\
$\Xarg{\psi}$ &$\in$& $\subformulaset{\varphi}$
&$\Rightarrow$&
$\psi \in \subformulaset{\varphi}$;\\
$\Uarg{\mu}{\eta} $ &$\in$& $ \subformulaset{\varphi}$
&$\Rightarrow$&
$\mu, \eta \in \subformulaset{\varphi}$;\\
$\Rarg{\mu}{\eta} $ &$\in$& $  \subformulaset{\varphi}$
&$\Rightarrow$&
$\mu, \eta \in \subformulaset{\varphi}$;\\
$\Earg{\pi}{x}{\psi}  $ &$\in$& $ \subformulaset{\varphi}$
&$\Rightarrow$&
$\psi \in \subformulaset{\varphi}$;\\
$\Aarg{\pi}{x}{\psi} $ &$\in$& $\subformulaset{\varphi}$
&$\Rightarrow$&
$\psi \in \subformulaset{\varphi}$.
\end{tabular}

The set $S$ additionally contains two distinct states, an accepting state $\accstate$ and its negation $\rejstate$. Both of these are ``pit'' states with any outgoing transition looping back to themselves.

A run of $A_\varphi$ is characterized by the states of $S$ that are visited as well as by the variable assignments that hold during these visits. Hence, the current ``real'' states of $A_\varphi$ can be thought as a pair $(p, \psi)$ from $\pfuncset{V}{\image{}} \times~S$ where $p$ represents the variable assignments that currently holds at $\psi$. The transition function $\rho : \pfuncset{V}{\image{}} \times~S \times $ also operates  on $\pfuncset{V}{\image{}} \times~S$ rather than on $S$. 

The initial state $s^0 \in S$ is $\varphi$ itself and the initial input of  $\rho$  is the couple $(\varnothing, s^0)$ where $\varnothing$ denotes the empty partial function from $V$ to $\image{}$. The set $F\subset S$ of accepting states includes $\accstate$ and every Release formula in $S$. The variable assignments do not affect the acceptance of the run: any couple $(p, \psi)$ is accepting if $\psi \in F$.  We  write $F_\psi$ to refer to the set of accepting states built from any formula $\psi$. More formally, if $x$ and $y$ are variables or constants, and if $\psi$, $\mu$, and $\eta$ are \LTLE{} formul\ae{}, the set of accepting states is recursively defined according to the following rules:
\begin{itemize}
\item $F_{(x = y)} = F_{(x \neq y)} = \lbrace\accstate\rbrace$;
\item $F_{(\Xarg{\psi})} = F_{(\Earg{\pi}{x}{\psi})} = F_{(\Aarg{\pi}{x}{\psi})} = F_\psi$;
\item $F_{(\mu \lor \eta)} = F_{(\mu \land \eta)} = F_{(\Uarg{\mu}{\eta})} = F_\mu \cup F_\eta$;
\item $F_{(\Rarg{\mu}{\eta})} = F_\mu \cup F_\eta \cup \lbrace\Rarg{\mu}{\eta}\rbrace$.
\end{itemize}

It remains only to define the transition function $\rho$. This is efficiently done by listing a small amount of general rules that must be applied recursively. Let $p$ be a partial function from $V$ to $\image{}$, and $m$ be any message in $\Sigma$. We set:
\begin{itemize}
\item $\transfunc{p}{\accstate}{m} =
(\varnothing, \accstate)$;
\item $\transfunc{p}{\rejstate}{m} =
(\varnothing, \rejstate)$;
\item $\transfunc{p}{x = y}{m} =
\begin{cases}
(\varnothing, \accstate) &
\text{if } p (x) = p (y) \\
(\varnothing, \rejstate) &
\text{otherwise;}
\end{cases}$
\item $\transfunc{p}{x \neq y}{m} =
\begin{cases}
(\varnothing, \accstate) &
\text{if } p (x) \neq p (y)
\\
(\varnothing, \rejstate) &
\text{otherwise;}
\end{cases}$
\item $\transfunc{p}{\mu \lor \eta}{m} =
\transfunc{p}{\mu}{m}
\lor
\transfunc{p}{\eta}{m}$;
\item $\transfunc{p}{\mu \land \eta}{m} =
\transfunc{p}{\mu}{m}
\land
\transfunc{p}{\eta}{m}$;
\item $\transfunc{p}{\Xarg{\psi}}{m} =
(p, \psi)$;
\item $\smash{
\transfunc{p}{\Uarg{\mu}{\eta}}{m} =
\transfunc{p}{\eta}{m}
\lor
\Big(
	\transfunc{p}{\mu}{m}
	\land
	(p, \Uarg{\mu}{\eta})
\Big)
}$;
\item $\smash{
\transfunc{p}{\Rarg{\mu}{\eta}}{m} =
\transfunc{p}{\mu \land \eta}{m}
\lor
\Big(
	\transfunc{p}{\eta}{m}
	\land
	(p, \Rarg{\mu}{\eta})
\Big)
}$;
\item $\smash{
\transfunc{p}{\Earg{\pi}{x}{\psi}}{m} =
\bigvee_{x_i \in Dom_m(\pi)}
\transfunc{p \cup \lbrace(x, x_i)\rbrace}{\psi}{m}
\lor
(\varnothing, \rejstate)
}$;
\item $\smash{
\transfunc{p}{\Aarg{\pi}{x}{\psi}}{m} =
\bigwedge_{x_i \in Dom_m(\pi)}
\transfunc{p \cup \lbrace(x, x_i)\rbrace}{\psi}{m}
\land
(\varnothing, \accstate)
}$.
\end{itemize}
If $Dom_m(\pi)$ is empty, we take the disjunction $\smash{\lor_{x_i \in Dom_m(\pi)}}$ in the case $\exists$ to be equivalent to $\FALSE$, and the conjunction $\smash{\land_{x_i \in Dom_m(\pi)}}$ in the case $\forall$ to be equivalent to $\TRUE$. Hence, the last two cases of the  transition formula evaluate to $(\varnothing, \rejstate)$ and $(\varnothing, \accstate)$ respectively. These special rules are logically consistent with the inherent meaning of $\exists$ and $\forall$. 

The constants $\TRUE$ and $\FALSE$, which we include in the set $\image{}$, may appear in the formul\ae{} of $S$ due to identities such as $\Farg{\psi} \equiv \Uarg{\TRUE}{\psi}$. In order for them to be compatible with the syntax of \LTLE{}, and our definition of $\rho$, we identify $\TRUE$ with the equality $\TRUE = \TRUE$, and $\FALSE$ with the inequality $\FALSE \neq \FALSE$.

Note that $\rho$ is undefined if its input contains an equality or inequality for which one of the variables are undefined in $p$. However, as long as $\varphi$, and by implication all of its subformulas are well-formed, this will never occur. Indeed, any variable in a well-formed \LTLE{} formula must be preceded by a quantifier on said variable. Therefore, $\rho$ will first process the quantifier and assign values to this variable before reaching the equality or inequality.

\begin{dfn} \label{dfnRunRho}
Let $\psi$ be an \LTLE{} formula, and $p : V \rightarrow \image{}$ be a partial function that assigns a value to every free variable in $\psi$. A \emph{run of $\rho$ on a message trace $\overline{m} = m_0, m_1, m_2, \ldots$  with root $(p, \psi)$} is an infinite $(\pfuncset{V}{\image{}} \times S)$-labelled tree that respects two conditions:
\begin{enumerate}
\item The root node must is labelled $(p, \psi)$;
\item Let's denote the distance between a node $N$ and the root by $d(N)$, its label by $l(N)$, and the labels of its $c_N$ children by $L_N = \lbrace l_1, \ldots, l_{c_N} \rbrace$. For any node $N$, the set $L_N$ must satisfy $\smash{\transfuncDeux{l(N)}{m_{d(N)}}}$, and for any $1 \leq i \leq c_N$, the label (or couple) $l_i$ must appear in $\smash{\transfuncDeux{l(N)}{m_{d(N)}}}$.
\end{enumerate}

For any \LTLE{} formula $\psi$, a run of $\rho$ is said to be \emph{$F_\psi$-accepting} if and only if every branch of the run has an infinite number of nodes whose labels contain a state in $F_\psi$; in other words, if every branch \emph{visits $F_\psi$ infinitely often}.
\end{dfn}

A  ``run of $\rho$''  is a tree of couples from $\pfuncset{V}{D} \times S$ generated by following the rules of $\rho$ for all messages in a trace. These  runs are more general than the runs of a modified alternating B\"uchi automaton $A_\varphi$ whose root is limited to the label $(\varnothing, \varphi)$. The previous definition allows us to properly name, and work with, parts of automaton runs, which happens frequently in Section \ref{sec:proof}. The parts that are themselves runs of $\rho$ will be called \emph{subruns}.

\begin{dfn}
Let $\psi$ be any \LTLE{} formula. A \emph{run of the automaton $A_\varphi$ on a message trace $\overline{m}$} is an run of $\rho$ on $\overline{m}$ and with root $(\varnothing, \varphi)$. Such a run is \emph{accepting} if and only if it is $F_\varphi$-accepting.  \emph{$A_\varphi$ accepts $\overline{m}$} if $A_\varphi$ admits at least one accepting run on $\overline{m}$.
\end{dfn}


\section{Proof of Correctness} \label{sec:proof}
\subsection{Preliminaries}
The following lemma shows that any transition formula in $\rho$ can be expressed as a disjunction of conjunctive clauses which themselves consist only in  equalities, inequalities,  couples in $\pfuncset{V}{\image{}} \times S$ or couples of the form $(\varnothing, \accstate)$ or $(\varnothing, \rejstate)$. This form can be obtained once every recursive rule of $\rho$, except the ones on equalities and inequalities, has been applied and makes it easier to identify the set of states that may compose the next level of a run tree. In what follows, for a given disjunction $D$,  the sets $E_d$ ranges over the equalities and inequalities, $N_d$ ranges over couples in $\pfuncset{V}{\image{}} \times S$ and $A_d$ ranges over couples of the form $(\varnothing, \accstate)$ or $(\varnothing, \rejstate)$. Intuitively, the disjunctive form represents the multiples combinations of  states that are simultaneously visited by the automata during a run.


\begin{lem} \label{lemRhoDecomp}
For any message $m \in M$, any \LTLE{} formula $\psi$, and any partial function $p : V \rightarrow \image{}$ that assigns every free variable in $\psi$, $\transfunc{p}{\psi}{m}$ is equal to an expression of the form
\begin{equation}
\bigvee_{d \in D}
\Big(
\bigwedge_{e \in E_d} \transfunc{p_{d, e}}{\varepsilon_{d, e}}{m}
\land
\bigwedge_{n \in N_d} (p_{d, n}, \psi_{d, n})
\land
\bigwedge_{a \in A_d} (\varnothing, \accstate_{d, a})
\Big)
\text{.}
\end{equation}
For all $d \in D$, $e \in E_d$, $n \in N_d$, and $a \in A_d$, $\varepsilon_{d, e}$ is an equality or inequality, $\psi_{d, n}$ is a subformula of $\psi$ (which includes $\psi$), $\accstate_{d, a}$ is either $\accstate$ or $\rejstate$, and $p_{d, e}$ and $p_{d, n}$ are partial functions from $V$ to $\image{}$ that assign every free variable in $\varepsilon_{d, e}$ and $\psi_{d, n}$ respectively.

Moreover, if $\timedepth{\psi} = 0$, then $N_d$ is empty for all $d \in D$.Otherwise, $N_d$ may not always be empty, and any $\psi_{d, n}$ fits only one of three descriptions:
\begin{itemize}
\item[1)] $\timedepth{\psi_{d, n}} < \timedepth{\psi}$;
\item[2)] $\timedepth{\psi_{d, n}} = \timedepth{\psi}$ and $\psi_{d, n} = \Uarg{\mu}{\eta}$ for some $\mu$, $\eta \in \subformulaset{\psi}$;
\item[3)] $\timedepth{\psi_{d, n}} = \timedepth{\psi}$ and $\psi_{d, n} = \Rarg{\mu}{\eta}$ for some $\mu$, $\eta \in \subformulaset{\psi}$.
\end{itemize}
\end{lem}

Observe that, since the rules governing the decomposition of equalities and inequalities are not appliqued when decomposing a formula in Normal form, if $\timedepth{\psi} = 0$, then $N_d$ and $A_d$ is empty for all $d \in D$. Indeed, elements are only be added in a set $A_d$ in the particular case when an a qualifier $\forall$ (resp. $\exists$) is encountered with an empty domain, leading to a trivial true (resp. false) verdict.

\begin{rmk}
For the sake of conciseness, we will regularly shorten the decomposition described in Lemma \ref{lemRhoDecomp} to:
\begin{equation*}
\transfunc{p}{\psi}{m} =
\bigvee_{d \in D}
\Big( \bigwedge_{c \in C_d} R_{d, c} \Big)
\text{.}
\end{equation*}

The set $C_d$, for all $d \in D$, includes every index in $E_d$, $N_d$, and $A_d$. The terms $R_{d, c}$,  are of the form $\transfunc{p_{d, e}}{\varepsilon_{d, e}}{m}$, $(p_{d, n}, \psi_{d, n})$, or $(\varnothing, \accstate_{d, a})$.
\end{rmk}

\begin{proof}
[Proof of Lemma \ref{lemRhoDecomp}]
Each term $R_{d, c}$ in the decomposition of $\transfunc{p}{\psi}{m}$ is obtained by recursively applying the rules of $\rho$. Every rule, if applicable, must be used exhaustively, except  the rules on equalities and inequalities, which are not applied at this point.  This restriction guarantees that  terms of the form $\transfunc{p_{d, e}}{\varepsilon_{d, e}}{m}$ are preserved. Thus, terms of the form $(\varnothing, \accstate)$ or $(\varnothing, \rejstate)$ only occur after evaluating the the quantifiers $\exists_{\pi}$ and $\forall_{\pi}$ when $Dom_m(\pi) = \varnothing$.

It follows from the restriction above, and the rules of $\rho$, that $\transfunc{p}{\psi}{m}$ can be recursively decomposed into a set consisting of couples $(f, \phi) \in \pfuncset{V}{I} \times S$ as well as of terms of the form $\transfunc{p_{d, e}}{\varepsilon_{d, e}}{m}$. Any output of $\rho$ that does not match one of these forms may be further decomposed. However, the recursive decomposition process cannot go on indefinitely, because while the rules of $\rho$ sometimes increase the number of terms, the output states always have a smaller depth. 

Moreover, a couple of the form  $(f, \phi)$ only appears if $\psi$ contains temporal operators or quantifiers. If $(f, \phi)$ is output by a rule related to the operator $\X$, $\U$ or $\R$, then the state $\phi$ is a subformula of $\psi$. Otherwise, if a couple $(f, \phi)$  is generated to a quantifier $\exists_{\pi}$ or $\forall_{\pi}$,  $f = \varnothing$ and $\phi \in \{ \accstate, \rejstate \}$. Hence, $(f, \phi)$ is either of the form $(p_{d, n}, \psi_{d, n})$ ,  $(\varnothing, \accstate_{d, a})$ or $(\varnothing, \rejstate_{d, a})$ and can thus be indexed by a set $E_d$, $N_d$, or $A_d$. It remains to show that these terms can be arranged into a disjunction of conjunctions.

We begin by showing that $\transfunc{p}{\psi}{m}$, and all of its partial decompositions, match a form that is similar, but not identical, to the desired one. It is still a disjunction of conjunctions, but its terms are partitioned by sets $T_d$ and $C_d = E_d \cup A_d \cup N_d$ instead of just $C_d$. More precisely, we consider the form
\begin{equation} \label{eqnRhoTempDecomp}
\bigvee_{d \in D}
\Big(
	\bigwedge_{t \in T_d} \transfunc{p_{d, t}}{\psi_{d, t}}{m}
	\land
	\bigwedge_{c \in C_d} R_{d, c}
\Big)
\text{.}
\end{equation}
For all $d \in D$, and $t \in T_d$, $\psi_{d, t}$ is a subformula of $\psi$ that is neither an equality nor inequality, and $p_{d, t}$ is a partial function from $V$ to $I$ that assigns every free variable in $\psi_{d, t}$. In short, the sets $T_d$ index terms that may be further decomposed. The terms $R_{d, c}$, which are either of the form $\transfunc{p_{d, e}}{\varepsilon_{d, e}}{m}$, $(p_{d, n}, \psi_{d, n})$,  $(\varnothing, \accstate_{d, a})$ or $(\varnothing, \rejstate_{d, a})$ , already satisfy the conditions of the lemma. In particular, the partial function of any term $R_{d, c}$ covers every free variable in its associated formula.

We prove the validity of this claim by induction. First, let's consider the initial term, $\transfunc{p}{\psi}{m}$, for which $p$ assigns every free variable in $\psi$. If $\psi$ is an equality or inequality, then this term is of the form $\transfunc{p_{d, e}}{\psi_{d, e}}{m}$. Otherwise, it matches the form $\transfunc{p_{d, t}}{\psi_{d, t}}{m}$. In both cases, $\transfunc{p}{\psi}{m}$ is the sole operand of a conjunction indexed by a set $T_d$ or $C_d$ where $D = \{ d \}$. Therefore, the base case holds.

For the induction step, we show that the partial decomposition of $\transfunc{p}{\psi}{m}$, possibly $\transfunc{p}{\psi}{m}$ itself also is of the form \eqref{eqnRhoTempDecomp}. If $T_d = \varnothing$ for every $d \in D$, then the decomposition only contains terms of the form $R_{d, c}$. Hence, it matches the expression claimed by this lemma, as well as the form \eqref{eqnRhoTempDecomp}. Otherwise, if $T_{d^*} \neq \varnothing$ for some $d^* \in D$, then it contains a term of the form $\transfunc{p_{d^*, t^*}}{\psi_{d^*, t^*}}{m}$ that is indexed by $d^*$ and some $t^* \in T_{d^*}$. Since $\transfunc{p_{d^*, t^*}}{\psi_{d^*, t^*}}{m}$ is further decomposable, we apply the corresponding rule of $\rho$ to it. We prove below that the resulting terms can always be arranged in a way that preserves the disjunctive form \eqref{eqnRhoTempDecomp}. Therefore, the induction step also holds.

The details of maintaining the disjunctive form depend on the nature of the formula $\psi_{d,^*, t^*}$. If it is a Next formula $\Xarg{\mu}$, then $\transfunc{p_{d^*, t^*}}{\psi_{d^*, t^*}}{m}$ outputs a couple $(p_{d^*, t^*}, \mu)$ that is simply added to the conjunction indexed by $N_{d^*} \subseteq C_{d^*}$. If $\psi_{d^*, t^*}$ is of the form $\Earg{\pi}{x}{\mu}$ or $\Aarg{\pi}{x}{\mu}$, but $Dom_m(\pi) = \varnothing$, then $\transfunc{p_{d^*, t^*}}{\psi_{d^*, t^*}}{m}$ outputs $(\varnothing, \rejstate)$ or $(\varnothing, \accstate)$ respectively. This output is then added to the conjunction indexed by $A_{d^*} \subseteq C_{d^*}$. The remaining cases, which concern the forms $\mu \lor \eta$, $\mu \land \eta$, $\Uarg{\mu}{\eta}$, and $\Rarg{\mu}{\eta}$, but also $\Earg{\pi}{x}{\mu}$ and $\Aarg{\pi}{x}{\mu}$ when $Dom_m(\pi) \neq \varnothing$, are covered below.

\begin{flushleft}
\textit{Case} $\psi_{d^*, t^*} = \Aarg{\pi}{x}{\mu}$ or $\psi_{d^*, t^*} = \mu \land \eta$ ($\mu, \eta \in \subformulaset{\psi}$, $Dom_m(\pi) \neq \varnothing$)
\end{flushleft}

Let's first suppose that $\psi_{d^*, t^*} = \Aarg{\pi}{x}{\mu}$. By applying the corresponding rule of $\rho$ on $\transfunc{p_{d^*, t^*}}{\psi_{d^*, t^*}}{m}$, and by considering the terms surrounding its output, we obtain the following equation:

\small
\begin{IEEEeqnarray*}{Rcl}
\IEEEeqnarraymulticol{3}{l}{
	\bigg(
		\transfunc{p_{d^*, t^*}}{\Aarg{\pi}{x}{\mu}}{m}
		\land \!\!
		\bigwedge_{\substack{t \in T_{d^*} \\ t \neq t^*}}
		\transfunc{p_{d^*, t}}{\psi_{d^*, t}}{m}
		\land \!\!
		\bigwedge_{c \in C_{d^*}}
		R_{d^*, c}
	\bigg)
	\lor
	\bigvee_{\substack{d \in D \\ d \neq d^*}}\ldots
} 
\\ = & \\
\IEEEeqnarraymulticol{2}{l}{
	\bigg(
		\bigwedge_{x_i \in Dom_m(\pi)}
		\transfunc{p_{d^*, t^*} \cup \{ (x, x_i) \}}{\mu}{m}
		\land \!\!
		\bigwedge_{\substack{t \in T_{d^*} \\ t \neq t^*}}
		\transfunc{p_{d^*, t}}{\psi_{d^*, t}}{m}
		\land \!\!
		\bigwedge_{c \in C_{d^*}}
		R_{d^*, c}
	\bigg)
	\lor
	\bigvee_{\substack{d \in D \\ d \neq d^*}}\ldots
}
\end{IEEEeqnarray*}
\normalsize
The term $\transfunc{p_{d^*, t^*}}{\psi_{d^*, t^*}}{m}$, which belongs in the conjunction indexed by $T_{d^*}$, is replaced by a conjunction of $|Dom_m(\pi)|$ terms. Every one of them can be indexed by either $T_{d^*}$ or $E_{d^*}$ depending on whether $\mu$ is an equality, an inequality, or neither. This is because the partial function $p_{d^*, t^*} \cup \{ (x, x_i) \}$, for any value $x_i$, assigns every free variable in $\mu$. The formulas $\mu$ and $\Aarg{\pi}{x}{\mu}$ share the same variables, with the only difference being that $x$ is not bound inside $\mu$. As such, the free variables of $\mu$ consists of $x$ and the free variables of $\psi_{d^*, t^*}$, which is exactly what $p_{d^*, t^*} \cup \{ (x, x_i) \}$ covers. Hence, the right-hand expression of the previous equation matches form \eqref{eqnRhoTempDecomp}.

The subcase $\psi_{d^*, t^*} = \mu \land \eta$, in which we decompose $\transfunc{p_{d^*, t^*}}{\psi_{d^*,t^*}}{m}$ into a conjunction of two terms rather than $|Dom_m(\pi)|$, is treated similarly. Obviously, the partial function $p_{d^*, t^*}$ assigns every free variable in $\mu$ and $\eta$.

\begin{flushleft}
\textit{Case} $\psi_{d^*, t^*} = \Earg{\pi}{x}{\mu}$ or $\psi_{d^*, t^*} = \mu \lor \eta$ ($\mu, \eta \in \subformulaset{\psi}$, $Dom_m(\pi) \neq \varnothing$)
\end{flushleft}

Let's first suppose that $\psi_{d^*, t^*} = \Earg{\pi}{x}{\mu}$. As with the subcase $\Aarg{\pi}{x}{\mu}$, $\rho$ decomposes the term $\transfunc{p_{d^*, t^*}}{\psi_{d^*, t^*}}{m}$ into $|Dom_m(\pi)|$ terms of the form $\transfunc{p_{d^*, t^*} \cup \{ (x, x_i) \}}{\mu}{m}$. However, unlike the subcase $\Aarg{\pi}{x}{\mu}$, these terms are joined by the operator $\lor$. If $|Dom_m(\pi)| \geq 2$, they create a disjunction inside the conjunction $d^* \in D$. As shown by the second member of the equation below, the resulting expression is not a conjunction of disjunctions.

\small
\begin{IEEEeqnarray*}{Rcl}
\IEEEeqnarraymulticol{3}{l}{
	\bigg(
		\transfunc{p_{d^*, t^*}}{\Earg{\pi}{x}{\mu}}{m}
		\land \!\!
		\bigwedge_{\substack{t \in T_{d^*} \\ t \neq t^*}}
		\transfunc{p_{d^*, t}}{\psi_{d^*, t}}{m}
		\land \!\!
		\bigwedge_{c \in C_{d^*}}
		R_{d^*, c}
	\bigg)
	\lor
	\bigvee_{\substack{d \in D \\ d \neq d^*}}\ldots
} 
\\ = &\\
\IEEEeqnarraymulticol{2}{l}{
	\bigg(
		\bigvee_{x_i \in Dom_m(\pi)}
		\transfunc{p_{d^*, t^*} \cup \{ (x, x_i) \}}{\mu}{m}
		\land \!\!
		\bigwedge_{\substack{t \in T_{d^*} \\ t \neq t^*}}
		\transfunc{p_{d^*, t}}{\psi_{d^*, t}}{m}
		\land \!\!
		\bigwedge_{c \in C_{d^*}}
		R_{d^*, c}
	\bigg)
	\lor
	\bigvee_{\substack{d \in D \\ d \neq d^*}}\ldots
}
\\ = &\\
\IEEEeqnarraymulticol{2}{l}{
	\bigvee_{x_i \in Dom_m(\pi)}
	\bigg(
		\transfunc{p_{d^*, t^*} \cup \{ (x, x_i) \}}{\mu}{m}
		\land \!\!
		\bigwedge_{\substack{t \in T_{d^*} \\ t \neq t^*}}
		\transfunc{p_{d^*, t}}{\psi_{d^*, t}}{m}
		\land \!\!
		\bigwedge_{c \in C_{d^*}}
		R_{d^*, c}
	\bigg)
	\lor
	\bigvee_{\substack{d \in D \\ d \neq d^*}}\ldots
}
\end{IEEEeqnarray*}
\normalsize

Therefore, we distribute the surrounding conjunctions, which are indexed by the sets $T_{d^*} \setminus \{ t^* \}$ and $C_{d^*}$, over the problematic disjunction. This results in $|Dom_m(\pi)|$ almost identical conjunctions that differ by the value $x_i$ inside their term $\transfunc{p_{d^*, t^*} \cup \{ (x, x_i) \}}{\mu}{m}$. As shown by the last member of the equation above, this new expression is a disjunction of conjunctions.

We established, with the subcase $\Aarg{\pi}{x}{\mu}$, that the partial function $p_{d^*, t^*} \cup \{ (x, x_i) \}$, for any value $x_i$, assigns every free variable in $\mu$. Thus, every term of the form $\transfunc{p_{d^*, t^*} \cup \{ (x, x_i) \}}{\mu}{m}$, depending on the nature of $\mu$, can be indexed by the set $T_{d^*}$ or $E_{d^*} \subseteq C_{d^*}$ associated with its conjunction. Hence, the expression obtained matches form \eqref{eqnRhoTempDecomp}.

The subcase $\psi_{d^*, t^*} = \mu \lor \eta$, in which we decompose $\transfunc{p_{d^*, t^*}}{\psi_{d^*,t^*}}{m}$ into a disjunction of two terms rather than $|Dom_m(\pi)|$, is treated similarly.

\begin{flushleft}
\textit{Case} $\psi_{d^*, t^*} = \Uarg{\mu}{\eta}$ or $\psi_{d^*, t^*} = \Rarg{\mu}{\eta}$ ($\mu, \eta \in \subformulaset{\psi}$)
\end{flushleft}

Let's first consider the case where $\psi_{d^*, t^*} = \Uarg{\mu}{\eta}$. By the definition of $\rho$  $\transfunc{p_{d^*, t^*}}{\psi_{d^*, t^*}}{m}$ is decomposable into a disjunction of two terms, namely  $\transfunc{p_{d^*, t^*}}{\eta}{m}$ and $\transfunc{p_{d^*, t^*}}{\mu}{m} \land (p_{d^*, t^*}, \Uarg{\mu}{\eta})$. It follows from the previous case that the conjunctions surrounding them must be distributed over the disjunction. The  process is displayed here:

\begin{IEEEeqnarray*}{Rcl}
\IEEEeqnarraymulticol{3}{l}{
	\bigg(
		\transfunc{p_{d^*, t^*}}{\Uarg{\mu}{\eta}}{m}
		\land \!\!
		\bigwedge_{\substack{t \in T_{d^*} \\ t \neq t^*}}
		\transfunc{p_{d^*, t}}{\psi_{d^*, t}}{m}
		\land \!\!
		\bigwedge_{c \in C_{d^*}}
		R_{d^*, c}
	\bigg)
	\lor
	\bigvee_{\substack{d \in D \\ d \neq d^*}}\ldots
} 
\\ = & 
\Bigg( & 
	\transfunc{p_{d^*, t^*}}{\eta}{m}
	\lor
	\Big(
		\transfunc{p_{d^*, t^*}}{\mu}{m}
		\land
		(p_{d^*, t^*}, \Uarg{\mu}{\eta})
	\Big) 
\\[-6pt] && 
	\land \!\!
	\bigwedge_{\substack{t \in T_{d^*} \\ t \neq t^*}}
	\transfunc{p_{d^*, t}}{\psi_{d^*, t}}{m}
	\land \!\!
	\bigwedge_{c \in C_{d^*}}
	R_{d^*, c}
\Bigg)
\lor
\bigvee_{\substack{d \in D \\ d \neq d^*}}\ldots 
\\[8pt] = & 
\IEEEeqnarraymulticol{2}{l}{
	\Bigg(
	\transfunc{p_{d^*, t^*}}{\eta}{m}
	\land \!\!
	\bigwedge_{\substack{t \in T_{d^*} \\ t \neq t^*}}
	\transfunc{p_{d^*, t}}{\psi_{d^*, t}}{m}
	\land \!\!
	\bigwedge_{c \in C_{d^*}}
	R_{d^*, c}
	\Bigg)
} 
\\[-8pt] & 
\lor & 
\Bigg(
	\transfunc{p_{d^*, t^*}}{\mu}{m}
	\land
	(p_{d^*, t^*}, \Uarg{\mu}{\eta})
	\land \!\!
	\bigwedge_{\substack{t \in T_{d^*} \\ t \neq t^*}}
	\transfunc{p_{d^*, t}}{\psi_{d^*, t}}{m}
	\land \!\!
	\bigwedge_{c \in C_{d^*}}
	R_{d^*, c}
\Bigg) 
\\[-8pt] & 
\lor & 
\bigvee_{\substack{d \in D \\ d \neq d^*}}\ldots 
\end{IEEEeqnarray*}

The partial function $p_{d^*, t^*}$ assigns every free variable inside the subformulas $\mu$ and $\eta$. Thus, the output terms $\transfunc{p_{d^*, t^*}}{\mu}{m}$ and $\transfunc{p_{d^*, t^*}}{\eta}{m}$, depending on the nature of $\mu$ and $\eta$, can be indexed by their respective copy of $T_{d^*}$ or $E_{d^*} \subseteq C_{d^*}$. The couple $(p_{d^*, t^*}, \Uarg{\mu}{\eta})$, by comparison, can be indexed by of $N_{d^*} \subseteq C_{d^*}$. Thus, the resulting expression is of the form \eqref{eqnRhoTempDecomp}.

The subcase $\psi_{d^*, t^*} = \Rarg{\mu}{\eta}$ is treated similarly since the decomposition of $\transfunc{p_{d^*, t^*}}{\Rarg{\mu}{\eta}}{m}$ is nearly identical. Since $p_{d^*, t^*}$ assigns every free variable in $\Rarg{\mu}{\eta}$, it does the same for $\eta$ and $\mu \land \eta$. Note though that while $\mu$ and $\eta$ are subformulas of $\psi$ because $\Rarg{\mu}{\eta} \in \subformulaset{\psi}$, this is usually not the case for $\mu \land \eta$. Therefore, the term $\transfunc{p_{d^*, t^*}}{\mu \land \eta}{m}$ cannot usually be indexed by the set $T_{d^*}$, but its decomposition $\transfunc{p_{d^*, t^*}}{\mu}{m} \land \transfunc{p_{d^*, t^*}}{\eta}{m}$ always can.

The above reasoning proves that any intermediate decomposition of the term $\transfunc{p}{\psi}{m}$  (with the rules for (in)equality omitted) will be of the form \eqref{eqnRhoTempDecomp}. However, we also showed that the induction step cannot be repeated indefinitely, since the depth of the terms monotonically decreases and the set $T_d$ eventually becomes empty for every conjunction $d \in D$. When this happens, the resulting equation is a conjunction of disjunction of the form  described in described in Lemma \ref{lemRhoDecomp}.
\showboxcond{} \end{proof}

A very similar decomposition can be applied to a predicate representing an \LTLE{} formula $\psi$. As expected, this notation follows the semantics of \LTLE{}.  By convention, if $Dom_m(\pi)$ is empty, then disjunctions and conjunctions indexed by $x_i \in Dom_m(\pi)$ default to $\FALSE$ and $\TRUE$ respectively. We keep this convention, but we also propose a cosmetic change that allows the above definition to better match the definition of $\rho$.

$\logicprop{p}{\accstate}{\overline{m}} = \TRUE$ and $\logicprop{p}{\rejstate}{\overline{m}} = \FALSE$ for any partial function $p : V \rightarrow \image{}$ and any trace $\overline{m}$. Hence, if $Dom_m(\pi)$ is empty, we have that $\logicprop{p}{\Earg{\pi}{x}{\psi}}{\overline{m}}$ and $\logicprop{p}{\Aarg{\pi}{x}{\psi}}{\overline{m}}$ are equivalent to $\logicprop{\varnothing}{\rejstate}{\overline{m}^2}$ and $\logicprop{\varnothing}{\accstate}{\overline{m}^2}$ respectively. This creates an equivalences between the predicate rules and the rules of $\rho$.  The decomposition of a predicate for a formula $\psi$ is given in lemma \ref{lemLogicDecomp}.
\begin{lem} \label{lemLogicDecomp}
For any message trace $\overline{m} = m_1, m_2, \ldots$, any \LTLE{} formula $\psi$, and any partial function $p : V \rightarrow \image{}$ that assigns every free variable in $\psi$, $\logicprop{p}{\psi}{\overline{m}}$ is equivalent to an expression of the form
\begin{equation}
\bigvee_{d \in D}
\Big(
	\bigwedge_{e \in E_d} \logicprop{p_{d, e}}{\varepsilon_{d, e}}{\overline{m}}
	\land
	\bigwedge_{n \in N_d} \logicprop{p_{d, n}}{\psi_{d, n}}{\overline{m}^2}i
	\land
	\bigwedge_{a \in A_d} \logicprop{\varnothing}{\accstate_{d, a}}{\overline{m}^2}
\Big)
\text{.}
\end{equation}
For all $d \in D$, $e \in E_d$, $n \in N_d$, and $a \in A_d$, $\varepsilon_{d, e}$ is an equality or inequality, $\psi_{d, n}$ is a subformula of $\psi$ (which includes $\psi$), $\accstate_{d, a}$ is either $\accstate$ or $\rejstate$, and $p_{d, e}$ and $p_{d, n}$ are partial functions from $V$ to $\image{}$ that assign every free variable in $\varepsilon_{d, e}$ and $\psi_{d, n}$ respectively.

Moreover, if $\timedepth{\psi} = 0$, then $N_d$ is empty for all $d \in D$. Otherwise, $N_d$ may not always be empty, and any $\psi_{d, n}$ fit only one of three descriptions:
\begin{itemize}
\item[1)] $\timedepth{\psi_{d, n}} < \timedepth{\psi}$;
\item[2)] $\timedepth{\psi_{d, n}} = \timedepth{\psi}$ and $\psi_{d, n} = \Uarg{\mu}{\eta}$ for some $\mu$, $\eta \in \subformulaset{\psi}$;
\item[3)] $\timedepth{\psi_{d, n}} = \timedepth{\psi}$ and $\psi_{d, n} = \Rarg{\mu}{\eta}$ for some $\mu$, $\eta \in \subformulaset{\psi}$.
\end{itemize}
\end{lem}

\begin{proof}
Completely analogous to the proof of Lemma \ref{lemRhoDecomp}.
\showboxcond{}\end{proof}

Lemmas \ref{lemRhoDecomp} and \ref{lemLogicDecomp} identify a very strong connection between the semantics of \LTLE{} and the transition function $\rho$. In fact, it should come as no surprise, at this point, that their decompositions can be paired in such a way that they match, term for term.
\begin{lem} \label{lemDecompEquiv}
For any message trace $\overline{m} = m, m_2, m_3, \ldots$, any \LTLE{} formula $\psi$, and any partial function $p : V \rightarrow \image{}$ that assigns every free variable in $\psi$, the terms $p_{d, e}$, $p_{d, n}$, $\varepsilon_{d, e}$, $\psi_{d, n}$, and $\accstate_{d, a}$ in the decompositions of $\transfunc{p}{\psi}{m}$ and $\logicprop{p}{\psi}{\overline{m}}$ can be labelled in such a way that they perfectly match across the two decompositions for any $d \in D$, $e \in E_d$, $n \in N_d$, and $a \in A_d$.
\end{lem}

\subsection{$A_\varphi$ accepts $\overline{m} \Rightarrow \overline{m} \models \varphi$}
The proof relies upon the following theorem, which is straightforwardly derived from the definitions of $\U$ and $\R$ and describes the shape of the tree resulting from a run of $\transfunc{p}{\psi}{\overline{m}}$ with $\psi = \U $ or $\psi=\R$.

\begin{lem} \label{lemSubrunUR}
Let $\psi$, $\mu$, and $\eta$ be \LTLE{} formul\ae and $p : V \rightarrow \image{}$ be a partial function that assigns every free variable in $\psi$. Suppose there exists an $F_\psi$-accepting run of $\rho$ on a trace $\overline{m}$ with root $(p, \psi)$. Let's denote this run by $\rhorun{p}{\psi}{\overline{m}}$.
\begin{itemize}
\item[1)] If $\psi = \Uarg{\mu}{\eta}$, then:
	\begin{itemize}
	\item[$\bullet$] $\exists \ j \geq 1$ : $\rho$ admits an $F_\eta$-accepting run $\rhorun{p}{\eta}{\overline{m}^j}$ with root $(p, \eta)$;
	\item[$\bullet$] $\forall \ 1 \leq i < j$, $\rho$ admits an $F_\mu$-accepting run $\rhorun{p}{\mu}{\overline{m}^i}$ with root $(p, \mu)$.
	\end{itemize}
\item[2)] If $\psi = \Rarg{\mu}{\eta}$, then either:
	\begin{itemize}
	\item[$\bullet$] $\forall \ i \geq 1$, $\rho$ admits an $F_\eta$-accepting run $\rhorun{p}{\eta}{\overline{m}^i}$ with root $(p, \eta)$; or
	\end{itemize}
	\begin{itemize}
	\item[$\bullet$] $\exists \ j \geq 1$ : $\rho$ admits an $F_{\mu \land \eta}$-accepting run $\rhorun{p}{\mu \land \eta}{\overline{m}^j}$ with root $(p, \mu \land \eta)$;
	\item[$\bullet$] $\forall \ 1 \leq i < j$, $\rho$ admits an $F_\eta$-accepting run $\rhorun{p}{\eta}{\overline{m}^i}$ on $\overline{m}^i$  with root $(p, \eta)$.
	\end{itemize}
\end{itemize}
\end{lem}

Informally, the above lemma states that, if the automata admits an accepting run over $\phi=\Uarg{\mu}{\eta}$ (resp. $\phi=\Rarg{\mu}{\eta}$) then it admits  accepting runs for $\mu$ and $\eta$ on such prefixes and/or suffixes of the trace $\overline{m}$ as to be consistent with the semantics of these operators.

\begin{proof}

Because $\U$ and $\R$ lead to similar proofs, we give the details for $\U$, and focus on the key differences that occur with $\R$.

\begin{flushleft}
\textit{1)} $\psi = \Uarg{\mu}{\eta}$
\end{flushleft}

The first step of the derivation of run $\rhorun{p}{\psi}{\overline{m}}$ is dictated by the rule
\begin{equation*}
\transfunc{p}{\Uarg{\mu}{\eta}}{m} =
\transfunc{p}{\eta}{m} \lor
\Big(
\transfunc{p}{\mu}{m} \land
(p, \Uarg{\mu}{\eta})
\Big)
\text{.}
\end{equation*}

It is clear from the premise that $p$ assigns every free variable in $\mu$ and $\eta$. Thus, we can apply Lemma \ref{lemRhoDecomp} to both $\transfunc{p}{\eta}{m}$ and $\transfunc{p}{\mu}{m}$, which yields:
\begin{equation} \label{eqnUntilDecomp}
\transfunc{p}{\Uarg{\mu}{\eta}}{m} =
\bigvee_{d \in D_\eta}
\Big(\bigwedge_{c \in C_d} R_{d, c}\Big)
\lor
\bigvee_{d \in D_\mu}
\Big(\bigwedge_{c \in C_d} R_{d, c} \land (p, \Uarg{\mu}{\eta})\Big)
\text{.}
\end{equation}

The sets $D_\eta$ and $D_\mu$ encompass the decomposition of  $\transfunc{p}{\eta}{m}$ and \\ $\transfunc{p}{\mu}{m}$ respectively.

By definition \ref{dfnRunRho}, in a run of $\rho$, the children of a node $(p, \Uarg{\mu}{\eta})$ must include every couple $R_{d, c} \in D_\eta \cup D_\mu$  in the right-hand side of  \eqref{eqnUntilDecomp} Therefore, only two scenarios are possible:
\begin{itemize}
\item[(i)] For any $i \geq 1$, there exists a distinguished index $d_i$ such that  message $m_i$ is in $D_\mu$;
\item[(ii)] For some $j \geq 1$, there exists a distinguished index $d_j$ such that message $m_j$ is in $D_\eta$, and for any $1 \leq i < j$, an index $d_i$ such that message $m_i$ is in $D_\mu$.
\end{itemize}

The first scenario would occur if $\mu$ is verified by every message in $\overline{m}$ while the second occurs if $\mu$ holds continuously until some point $j$, where $\eta$ holds (these scenarios are not mutually exclusive because two or more indexes can be represented by the child nodes, but cover every possibility). Indeed, either some $j \geq 1$ linked to $D_\eta$ exists, in which case, the smallest possible $j$ satisfies (ii); or it does not. In that case, (i) holds. Note also that every time a chosen index is in $D_\mu$, a node $(p, \Uarg{\mu}{\eta})$ appears as a child, so $D_\mu$ or $D_\eta$ must be picked for the next message as well.

It turns out that the run $\rhorun{p}{\psi}{\overline{m}}$ follows scenario (ii). If it did not, then according to (i), and the previous discussion, there would be a branch in $\rhorun{p}{\psi}{\overline{m}}$ whose nodes are always $(p, \Uarg{\mu}{\eta})$. Since $\psi = \Uarg{\mu}{\eta}$ is not in $F_\psi$, the run would not be $F_\psi$-accepting, thus contradicting our premise.
Scenario (ii) echoes the desired results in regard to values $i, j$ and formulas $\mu, \eta$. It remains to extract, for any relevant value $i$ or $j$, the desired subrun from $\rhorun{p}{\psi}{\overline{m}}$. We only cover the details for $i$, as both cases lead to almost identical proofs.

Since index $d_i$ is in $D_\mu$, the children of a node $(p, \Uarg{\mu}{\eta})$ linked to message $m_i$ can be of three ``types''. They can be either $(p, \Uarg{\mu}{\eta})$ (type 1, mandatory), a couple $R_{d, c}$ in equation \eqref{eqnUntilDecomp} where $d = d_i$ (type 2, mandatory), or a couple $R_{d, c}$ in \eqref{eqnUntilDecomp} where $d \neq d_i$ (type 3, optional). To make sense of type 3, remember that an accepting run may include extra nodes as long as it respects $\rho$ and $F$. Now, observe that the union of all children of type 2, denoted $U_{2}^i$, satisfies the conjunction $\land_{c \in C_{d_i}}$ in \eqref{eqnUntilDecomp} if $m = m_i$. Thus, $U_{2}^i$ satisfies $\transfunc{p}{\mu}{m_i}$. Hence, if we take a root node $(p, \mu)$, set $U_{2}^i$ as its children, and keep every subrun of $\rhorun{p}{\psi}{\overline{m}}$ whose root is in $U_{2}^i$, we get a run $\rhorun{p}{\mu}{\overline{m}^i}$ starting at $m_i$ and with root $(p, \mu)$. 
We must now prove that $\rhorun{p}{\mu}{\overline{m}^i}$ is $F_\mu$-accepting. Since the original run is $F_{(\Uarg{\mu}{\eta})}$-accepting, all of its branches visit $F_{(\Uarg{\mu}{\eta})}$ infinitely often. This remains true for every branch in $\rhorun{p}{\mu}{\overline{m}^i}$ because they are infinite suffixes of branches in the original run ($i$ is finite). However, the formulas in $\rhorun{p}{\mu}{\overline{m}^i}$ are limited to subformulas of $\mu$ (Lemma \ref{lemRhoDecomp}), so what is visited infinitely often is actually $F_{(\Uarg{\mu}{\eta})} \cap (\subformulaset{\mu} \cup \lbrace\accstate\rbrace) = F_\mu$.

As we mentioned earlier, the proof for $j$ is very similar. The set $U_{2}^j$, where $d_j \in D_\eta$, identifies the desired subrun. The only notable difference is that a node $(p, \Uarg{\mu}{\eta})$, for message $m_j$, is not required to have a child of type 1.

\begin{flushleft}
\textit{2)} $\psi = \Rarg{\mu}{\eta}$
\end{flushleft}

The beginning of the run $\rhorun{p}{\psi}{\overline{m}}$ is dictated by the rule
\begin{equation*}
\transfunc{p}{\Rarg{\mu}{\eta}}{m} =
\transfunc{p}{\mu \land \eta}{m} \lor
\Big(
\transfunc{p}{\eta}{m} \land
(p, \Rarg{\mu}{\eta})
\Big)
\text{,}
\end{equation*}

which also applies each time the couple $(p, \Rarg{\mu}{\eta})$ appears in the run. As with $\U$, we can apply Lemma \ref{lemRhoDecomp} and distribute $(p, \Rarg{\mu}{\eta})$ over the decomposition of $\transfunc{p}{\eta}{m}$ to get
\begin{equation} \label{eqnReleaseDecomp}
\transfunc{p}{\Rarg{\mu}{\eta}}{m} =
\bigvee_{d \in D_{(\mu \land \eta)}}
\Big(\bigwedge_{c \in C_d} R_{d, c}\Big)
\lor
\bigvee_{d \in D_\eta}
\Big(\bigwedge_{c \in C_d} R_{d, c} \land (p, \Rarg{\mu}{\eta})\Big)
\end{equation}

where the set $D_{(\mu \land \eta)}$ encompass the decomposition of $\transfunc{p}{\mu \land \eta}{m}$.

The scenarios (i) and (ii) introduced for $\Uarg{\mu}{\eta}$ also apply in this case, but with minor adaptations regarding the index sets:
\begin{itemize}
\item[(i)] For any $i \geq 1$, there exists a distinguished index   $d_i$ such that  message $m_i$ is in $D_\eta$;
\item[(ii)] For some $j \geq 1$, there exists a distinguished index  $d_j$ such that message $m_j$ is in $D_{(\mu \land \eta)}$, and for any $1 \leq i < j$, there exists a distinguished index   $d_i$ such that message $m_i$ is in $D_\eta$.
\end{itemize}

However, by comparison to $\Uarg{\mu}{\eta}$, it is possible for $\rhorun{p}{\psi}{\overline{m}}$ to follow scenario (i). This is because $\psi$, this time, is included in $F_\psi$ due to being a Release formula. Until formulas, by comparison, are not accepting. Thus, even though (i) implies the existence of a branch in $\rhorun{p}{\psi}{\overline{m}}$ whose nodes are always $(p, \Rarg{\mu}{\eta})$, such a branch visits $F_\psi$ infinitely often and does not contradict our premise that the run is $F_\psi$-accepting. In fact, (i) and (ii) together echo the result stated in this lemma. The remaining of the proof for $i$ is analogous to the one for $i$ with $\U$. The same goes for $j$.

While every branch of the derived runs $\rhorun{p}{\mu \land \eta}{\overline{m}^j}$ and $\rhorun{p}{\eta}{\overline{m}^i}$ visits $F_\psi$ infinitely often, their formulas are restricted to subformulas of $\mu \land \eta$ and $\eta$ respectively (Lemma \ref{lemRhoDecomp}). Thus, the sets that are visited infinitely often are actually $F_\psi \cap (\subformulaset{\mu \land \eta} \cup \lbrace\accstate\rbrace) = F_{(\mu \land \eta)}$ and $F_\psi \cap (\subformulaset{\eta} \cup \lbrace\accstate\rbrace) = F_\eta$.
\showboxcond{} \end{proof}

We can now state the main lemmas of correction and completeness, indicating that an automata $A_\varphi$  admits an accepting run for sequence $\psi$ iff $\overline{m} \models \psi$.

\begin{lem} \label{lemRunToPredicate}
For any \LTLE{} formula $\psi$ and any partial function $p : V \rightarrow \image{}$ that assigns every free variable in $\psi$, if there exists an $F_\psi$-accepting run of $\rho$ on a trace $\overline{m} = m_1, m_2, \ldots$ and with root $(p, \psi)$, then $\logicprop{p}{\psi}{\overline{m}}$ is true.
\end{lem}

\begin{proof}
We proceed by strong induction on the temporal depth of \LTLE{} formul\ae. Recall that by Lemmas \ref{lemRhoDecomp} and \ref{lemLogicDecomp},  the decomposition of $\transfunc{p}{\psi}{m_1}$ and $\logicprop{p}{\psi}{\overline{m}}$ can be given in normal disjunctive form as follows and that Lemma \ref{lemDecompEquiv} assures us that for every $d \in D$, the terms of these two formulas can be matched to one another.

\begin{equation} \label{eqnRhoDecompLem}
\transfunc{p}{\psi}{m_1} =
\bigvee_{d \in D}
\Big(
\bigwedge_{e \in E_d} \transfunc{p_{d, e}}{\varepsilon_{d, e}}{m_1}
\land
\bigwedge_{n \in N_d} (p_{d, n}, \psi_{d, n})
\land
\bigwedge_{a \in A_d} (\varnothing, \accstate_{d, a})
\Big)
\end{equation}

$\logicprop{p}{\psi}{\overline{m}} \Leftrightarrow$\\
\begin{equation} \label{eqnLogicDecompLem}
\bigvee_{d \in D}
\Big(
\bigwedge_{e \in E_d} \logicprop{p_{d, e}}{\varepsilon_{d, e}}{\overline{m}}
\land
\bigwedge_{n \in N_d} \logicprop{p_{d, n}}{\psi_{d, n}}{\overline{m}^2}
\land
\bigwedge_{a \in A_d} \logicprop{\varnothing}{\accstate_{d, a}}{\overline{m}^2}
\Big)
\end{equation}


\begin{flushleft}
\textit{Base case}: Lemma \ref{lemRunToPredicate} holds for any formula of depth 0.
\end{flushleft}
Let $\timedepth{\psi} = 0$. This case is covered by Lemmas \ref{lemRhoDecomp} and \ref{lemLogicDecomp}, which state that $N_d$ and $A_d$ are empty for all $d \in D$. As a result, the  equations from  \ref{lemRhoDecomp} and \ref{lemLogicDecomp} can be rewritten as:
\begin{equation} \label{eqnRhoDecompLemBase}
\transfunc{p}{\psi}{m_1} =
\bigvee_{d \in D}
\Big(
\bigwedge_{e \in E_d} \transfunc{p_{d, e}}{\varepsilon_{d, e}}{m_1}
\Big)
\end{equation}
\begin{equation} \label{eqnLogicDecompLemBase}
\logicprop{p}{\psi}{\overline{m}} \Leftrightarrow
\bigvee_{d \in D}
\Big(
\bigwedge_{e \in E_d} \logicprop{p_{d, e}}{\varepsilon_{d, e}}{\overline{m}}
\Big)
\text{.}
\end{equation}

Since every formula $\varepsilon_{d, e}$ is either an equality or an inequality, it follows from the definition of $\rho$ that $\transfunc{p_{d, e}}{\varepsilon_{d, e}}{m_1}$ can only output $(\varnothing, \accstate)$ or $(\varnothing, \rejstate)$ and that the couple $(\varnothing, \accstate)$ is output if and only if $\logicprop{p_{d, e}}{\varepsilon_{d, e}}{\overline{m}}$ is true.

Since the only successor of a node $(\varnothing, \accstate)$ is itself, and since $\accstate$ is in $F_\psi$, this node generates an accepting branch. A node $(\varnothing, \rejstate)$ also loops back on itself, but since $\rejstate$ is not in $F_\psi$, the resulting branch is not accepting. Thus, the fact that there exists a $F_\psi$-accepting run $\rhorun{p}{\psi}{\overline{m}}$ implies that there exists a $d^* \in D$ for which $\logicprop{p_{d^*, e}}{\varepsilon_{d^*, e}}{\overline{m}}$ holds for all $e \in E_{d^*}$. The right-hand side of \eqref{eqnLogicDecompLemBase} holds as a result, and lemma \ref{lemRunToPredicate} holds in the base case.

\begin{flushleft}
\textit{Induction step}: For some natural $t > 0$, if Lemma \ref{lemRunToPredicate} holds for any formula of depth less than  or equal to $t$, the it holds for depth $t$.
\end{flushleft}

Let $\timedepth{\psi} = t$. The sets $N_d$ and $A_d$ may not be empty for all $d \in D$. By the definition of a run of $\rho$, the children of the root $(p, \psi)$ in the accepting run $\rhorun{p}{\psi}{\overline{m}}$ satisfy the right-hand side of equation \eqref{eqnRhoDecompLem}. Hence, there exists a $d^* \in D$ for which every output $\transfunc{p_{d^*, e}}{\varepsilon_{d^*, e}}{m_1}$ and every couple $(p_{d^*, n}, \psi_{d^*, n})$  or $(\varnothing, \accstate_{d, a})$ in \eqref{eqnRhoDecompLem} is a child of the root. From the base case, we can conclude that $\logicprop{p_{d^*, e}}{\varepsilon_{d^*, e}}{\overline{m}}$ is true for any $e \in E_{d^*}$. It remains to show that $\logicprop{p_{d^*, n}}{\psi_{d^*, n}}{\overline{m}^2}$ is true for any $n \in N_{d^*}$ and every $\logicprop{\varnothing}{\accstate_{d, a}}{\overline{m}^2}$ is true for any $a \in A_{d^*}$.

 We consider each case in turn.

A couple $(p_{d^*, n}, \psi_{d^*, n})$ is the root of at least one subrun in $\rhorun{p}{\psi}{\overline{m}}$ starting at message $m_2$. Let's denote it by $\rhorun{p_{d^*, n}}{\psi_{d^*, n}}{\overline{m}^2}$. We know from Lemma \ref{lemRhoDecomp} that $p_{d^*, n}$ assigns every free variable in $\psi_{d^*, n}$. We can also argue that this subrun is $F_{\psi_{d^*, n}}$-accepting since every branch in $\rhorun{p_{d^*, n}}{\psi_{d^*, n}}{\overline{m}^2}$ is an infinite suffix of a branch in $\rhorun{p}{\psi}{\overline{m}}$. Thus, $F_\psi$ is still visited infinitely often, but the formulas are limited to $\subformulaset{\psi_{d^*, n}}$ (Lemma \ref{lemRhoDecomp}). As a result, the intersection $F_\psi \cap (\subformulaset{\psi_{d^*, n}} \cup \lbrace\accstate\rbrace) = F_{\psi_{d^*, n}}$ is visited infinitely often which satisfies the antecedent of  Lemma \ref{lemRunToPredicate}.

The rest of the argument is made easy by the induction hypothesis. Indeed, by Lemma \ref{lemRhoDecomp}, a formula $\psi_{d^*, n}$ fits only one of three  possible descriptions:

\begin{flushleft}
1) $\timedepth{\psi_{d^*, n}} < t$
\end{flushleft}

The induction hypothesis can be applied directly to $\rhorun{p_{d^*, n}}{\psi_{d^*, n}}{\overline{m}^2}$. Therefore, $\logicprop{p_{d^*, n}}{\psi_{d^*, n}}{\overline{m}^2}$ holds.

\begin{flushleft}
2) $\timedepth{\psi_{d^*, n}} = t$ and $\psi_{d^*, n} = \Uarg{\mu}{\eta}$ where $\mu, \eta \in \subformulaset{\psi}$
\end{flushleft}

According to the case $\U$ of Lemma \ref{lemSubrunUR}, $\rhorun{p_{d^*, n}}{\Uarg{\mu}{\eta}}{\overline{m}^2}$ implies that:
\begin{itemize}
\item[$\bullet$] $\exists \ j \geq 2$ for which $\rho$ admits an $F_\eta$-accepting run   with root $(p_{d^*, n}, \eta)$;
\item[$\bullet$] $\forall \ 2 \leq i < j$, $\rho$ admits an $F_\mu$-accepting run  with root $(p_{d^*, n}, \mu)$.
\end{itemize}

Since $\timedepth{\mu}$ and $\timedepth{\eta}$ are less than $\timedepth{\Uarg{\mu}{\eta}} = t$, we can apply the induction hypothesis to the runs listed above:
\begin{itemize}
\item $\exists \ j \geq 2$ for which $\logicprop{p_{d^*, n}}{\eta}{\overline{m}^j}$ holds;
\item $\forall \ 2 \leq i < j$, $\logicprop{p_{d^*, n}}{\mu}{\overline{m}^i}$ holds.
\end{itemize}

By definition of the $\U$ operator, $\logicprop{p_{d^*, n}}{\Uarg{\mu}{\eta}}{\overline{m}^2}$ holds.

\begin{flushleft}
3) $\timedepth{\psi_{d^*, n}} = t$ and $\psi_{d^*, n} = \Rarg{\mu}{\eta}$ where $\mu, \eta \in \subformulaset{\psi}$
\end{flushleft}

The proof is analogous to the one for 2). Upon using Lemma \ref{lemSubrunUR} and the induction hypothesis, we get the following possibilities:
\begin{itemize}
\item[$\bullet$] $\forall \ i \geq 2$, $\logicprop{p_{d^*, n}}{\eta}{\overline{m}^i}$holds;
\end{itemize}
or
\begin{itemize}
\item[$\bullet$] $\exists \ j \geq 2$ for which $\logicprop{p_{d^*, n}}{\mu \land \eta}{\overline{m}^j}$ holds;
\item[$\bullet$] $\forall \ 2 \leq i < j$, $\logicprop{p_{d^*, n}}{\eta}{\overline{m}^i}$ holds.
\end{itemize}

By the definition of the $\R$ operator, $\logicprop{p_{d^*, n}}{\Rarg{\mu}{\eta}}{\overline{m}^2}$ holds in both cases.

We now turn to the set $A_{d^*}$. Let $\logicprop{\varnothing}{\accstate_{d^*, a}}{\overline{m}}$ be state in $A_d^*$. Since the run $\rhorun{p}{\psi}{\overline{m}}$ which generates this state is accepting, $\logicprop{\varnothing}{\accstate_{d^*, a}}{\overline{m}}$ is necessarily of the form $\logicprop{\varnothing}{\accstate{}}{\overline{m}}$.  By the definition of the transition function, this state is reached  after the application of the transition function to a formula of the form $\psi=  \forall _\pi: \psi'$ for which the $Dom_m(\pi)$ is empty, for some message $\overline{m}'$ and valuation function $p$. Since  $\timedepth{\psi} < t$, we can conclude form the induction hypothesis that $\logicprop{p}{\psi}{\overline{m}'}$ holds.

In conclusion, since every term $\logicprop{p_{d, e}}{\varepsilon_{d, e}}{\overline{m}}$ and $\logicprop{p_{d, n}}{\psi_{d, n}}{\overline{m}^2}$ in \eqref{eqnLogicDecompLem} holds for some $d \in D$, it follows that $\logicprop{p}{\psi}{\overline{m}}$ holds.
\showboxcond{}
\end{proof}

\begin{prs} \label{prpRunToPredicate}
For any \LTLE{} formula $\varphi$ devoid of free variables, and any message trace $\overline{m}$, if the automaton $A_\varphi$ accepts $\overline{m}$, then $\overline{m}$ satisfies $\varphi$.
\end{prs}

\begin{proof}
An accepting run of $A_\varphi$ on $\overline{m}$ is, by definition, an $F_\varphi$-accepting run of $\rho$ with root $(\varnothing, \varphi)$. By Lemma \ref{lemRunToPredicate}, $\rhorun{\varnothing}{\psi}{\overline{m}}$ is accepting implies that  $\logicprop{\varnothing}{\varphi}{\overline{m}}$ is true, which in turns means that ``$\overline{m}$ satisfies $\varphi$ given the free variable assignments in $\varnothing$''.
\showboxcond{}\end{proof}

\subsection{$\overline{m} \models \varphi \Rightarrow A_\varphi$ accepts $\overline{m}$}
The following lemma states the connection between the semantics of the the temporal operator $\R$ and $\U$ and the tree decomposition of these formul\ae{}. It is analogous to Lemma \ref{lemSubrunUR}. 

\begin{lem} \label{lemSuperrunUR}
Let $\mu$ and $\eta$ be \LTLE{} formulas, and $p : V \rightarrow D$ be a partial function that assigns every free variable in $\mu$ and $\eta$. Let $\overline{m} = m_1, m_2, \ldots$ be a message trace. If:
\begin{itemize}
\item[$\bullet$] $\forall \ i \geq 1$, $\rho$ admits an $F_\eta$-accepting run $\rhorun{p}{\eta}{\overline{m}^i}$  with root $(p, \eta)$,
\end{itemize}
then $\rho$ admits an $F_{(\Rarg{\mu}{\eta})}$-accepting run $\rhorun{p}{\Rarg{\mu}{\eta}}{\overline{m}}$ with root $(p, \Rarg{\mu}{\eta})$. Furthermore, if:
\begin{itemize}
\item[$\bullet$] $\exists \ j \geq 1$ for which $\rho$ admits an $F_{(\mu \land \eta)}$-accepting run $\rhorun{p}{\mu \land \eta}{\overline{m}^j}$  with root $(p, \mu \land \eta)$ (resp. $F_\eta$-accepting run $\rhorun{p}{\eta}{\overline{m}^j}$  with root $(p, \eta)$);
\item[$\bullet$] $\forall \ 1 \leq i < j$, $\rho$ admits an $F_\eta$-accepting run $\rhorun{p}{\eta}{\overline{m}^i}$  with root $(p, \eta)$ (resp. $F_\mu$-accepting run $\rhorun{p}{\mu}{\overline{m}^i}$  with root $(p, \mu)$),
\end{itemize}
then $\rho$ admits an $F_{(\Rarg{\mu}{\eta})}$-accepting run $\rhorun{p}{\Rarg{\mu}{\eta}}{\overline{m}}$   with root $(p, \Rarg{\mu}{\eta})$ (resp. $F_{(\Uarg{\mu}{\eta})}$-accepting run $\rhorun{p}{\Uarg{\mu}{\eta}}{\overline{m}}$  with root $(p, \Uarg{\mu}{\eta})$).
\end{lem}

\begin{proof}

\begin{flushleft}
1) Case `` $\forall \ i \geq 1$, $\rho$ admits an $F_\eta$-accepting run $\rhorun{p}{\eta}{\overline{m}^i}$ on $\overline{m}^i$ and with root $(p, \eta)$''
\end{flushleft}

Let's  suppose that $\rho$ admits an $F_\eta$-accepting run $\rhorun{p}{\eta}{\overline{m}^i}$ for all $i \geq 1$. We must build an $F_{(\Rarg{\mu}{\eta})}$-accepting run $\rhorun{p}{\Rarg{\mu}{\eta}}{\overline{m}}$ that includes each one of these runs.

The beginning of any run $\rhorun{p}{\eta}{\overline{m}^i}$, and of the run $\rhorun{p}{\Rarg{\mu}{\eta}}{\overline{m}}$, is dictated by the rule
\begin{equation*}
\transfunc{p}{\Rarg{\mu}{\eta}}{m} =
\transfunc{p}{\mu \land \eta}{m} \lor
\Big(
\transfunc{p}{\eta}{m} \land
(p, \Rarg{\mu}{\eta})
\Big)
\text{,}
\end{equation*}
which is also applied each time the children of a node $(p, \Rarg{\mu}{\eta})$ are considered. Since $p$ assigns every free variable in $\mu$ and $\eta$, we can use Lemma \ref{lemRhoDecomp} as we did in the proof of Lemma \ref{lemSubrunUR} to obtain the rule
\begin{equation} \label{eqnReleaseDecompSuper}
\transfunc{p}{\Rarg{\mu}{\eta}}{m} =
\bigvee_{d \in D_{(\mu \land \eta)}}
\Big(\bigwedge_{c \in C_d} R_{d, c}\Big)
\lor
\bigvee_{d \in D_\eta}
\Big(\bigwedge_{c \in C_d} R_{d, c} \land (p, \Rarg{\mu}{\eta})\Big)
\text{.}
\end{equation}
The sets $D_{(\mu \land \eta)}$ and $D_\eta$ encompasses the decomposition of $\transfunc{p}{\mu \land \eta}{m}$ and $\transfunc{p}{\eta}{m}$ respectively.

By Definition \ref{dfnRunRho}, in a run $\rhorun{p}{\eta}{\overline{m}^i}$, there exits at least one $d_{\eta}^{\!\; i} \in D_\eta$ for which any couple $\smash{R_{d_{\eta}^{i}, c}}$ in \eqref{eqnReleaseDecompSuper} is a child of the root (if $\smash{c \in E_{d_{\eta}^{i}}}$, the ``couple $\smash{R_{d_{\eta}^{i}, c}}$'' is the output of $\smash{R_{d_{\eta}^{i}, c}}$). Let's denote the set of all couples $\smash{R_{d_{\eta}^{i}, c}}$ by $\smash{R_{d_{\eta}^i}}$. Similarly, for each node $(p, \Rarg{\mu}{\eta})$ in the run $\rhorun{p}{\Rarg{\mu}{\eta}}{\overline{m}}$, there must be at least one $d^* \in D_{(\mu \land \eta)} \cup D_\eta$ for which any couple in \eqref{eqnReleaseDecompSuper} indexed by $d^*$ is a child of $(p, \Rarg{\mu}{\eta})$.

Let's consider the case where $i = 1$. If we set $\smash{d^* = d_{\eta}^1}$ for the root of the node $\rhorun{p}{\Rarg{\mu}{\eta}}{\overline{m}}$, then \eqref{eqnReleaseDecompSuper} implies that the children must at least include the set $\smash{R_{d_{\eta}^1}}$ and a node $(p, \Rarg{\mu}{\eta})$. We choose not to include others nodes, so it remains to define the subruns generated by $\smash{R_{d_{\eta}^1}}$ and $(p, \Rarg{\mu}{\eta})$. The subruns generated by the former are easy because they can be copied from $\rhorun{p}{\eta}{\overline{m}^1}$. Since this run is $F_\eta$-accepting, we know that every branch in a copied subrun visits $F_\eta$ infinitely often. With the inclusion $F_\eta \subset F_{(\Rarg{\mu}{\eta})}$, we can also state that every branch visits $F_{(\Rarg{\mu}{\eta})}$ infinitely often. As for the child node $(p, \Rarg{\mu}{\eta})$, its subrun is obtained by repeating the previous procedure for $i = 2$ and beyond. We use $d_{\eta}^{\!\; i}$ for every $i \geq 2$, and we copy subruns in $\rhorun{p}{\eta}{\overline{m}^i}$ for $\smash{R_{d_{\eta}^i}}$.

The resulting run of $\rho$ on $\overline{m}$ and with root $(p, \Rarg{\mu}{\eta})$ admits two types of branches. All but one eventually reach a node $\smash{R_{d_{\eta}^{i}, c}}$ for some $i \geq 1$, and thus visit $F_{(\Rarg{\mu}{\eta})}$ infinitely often. The remaining branch never reaches a node $\smash{R_{d_{\eta}^{i}, c}}$, and thus only visits the node $(p, \Rarg{\mu}{\eta})$. Fortunately, the formula $\Rarg{\mu}{\eta}$ identifies an accepting state, so this branch also visits $F_{(\Rarg{\mu}{\eta})}$ infinitely often. These last observations make the resulting run $F_{(\Rarg{\mu}{\eta})}$-accepting.

\begin{flushleft}
2) Case : `` $\exists \ j \geq 1$ for which $\rho$ admits an $F_{(\mu \land \eta)}$-accepting run $\rhorun{p}{\mu \land \eta}{\overline{m}^j}$ on $\overline{m}^j$ and with root $(p, \mu \land \eta)$ (resp. $F_\eta$-accepting run $\rhorun{p}{\eta}{\overline{m}^j}$ on $\overline{m}^j$ and with root $(p, \eta)$);\\
 $\forall \ 1 \leq i < j$, $\rho$ admits an $F_\eta$-accepting run $\rhorun{p}{\eta}{\overline{m}^i}$ on $\overline{m}^i$ and with root $(p, \eta)$ (resp. $F_\mu$-accepting run $\rhorun{p}{\mu}{\overline{m}^i}$ on $\overline{m}^i$ and with root $(p, \mu)$)''.
\end{flushleft}

Let's now suppose that $\rho$ admits an $F_{(\mu \land \eta)}$-accepting run $\rhorun{p}{\mu \land \eta}{\overline{m}^j}$ for some $j \geq 1$. Let's also suppose that for any $1 \leq i < j$, $\rho$ admits an $F_\eta$-accepting run $\rhorun{p}{\eta}{\overline{m}^i}$. As with the first ``case'', we must build an $F_{(\Rarg{\mu}{\eta})}$-accepting run $\rhorun{p}{\Rarg{\mu}{\eta}}{\overline{m}}$ that includes all these runs.

Until message $m_j$ in $\overline{m}$ is reached, we can use the approach described in the first ``case'' to build a partial run from the root $(p, \Rarg{\mu}{\eta})$. Thus, for all $1 \leq i < j$, a node $(p, \Rarg{\mu}{\eta})$ in our partial run, upon reading message $m_i$, is followed by subruns in $\rhorun{p} {\eta}{\overline{m}^i}$ and by a node $(p, \Rarg{\mu}{\eta})$. It remains to define a subrun on the trace $\overline{m}^j$ generated by a node $(p, \Rarg{\mu}{\eta})$.

By Definition \ref{dfnRunRho}, in the run $\rhorun{p}{\mu \land \eta}{\overline{m}^j}$, there is an index $d^{\!\; j} \in D_{(\mu \land \eta)}$ for which a couple $R_{d, c}$ in \eqref{eqnReleaseDecompSuper} is a child of the root $(p, \mu \land \eta)$ if $d = d^{\!\; j}$. Let's denote the set of all couples $\smash{R_{d^j, c}}$ by $\smash{R_{d^j}}$. Not only does $\smash{R_{d^j}}$ satisfies the right-hand side of \eqref{eqnReleaseDecompSuper}, the subrun of any couple $\smash{R_{d^j, c}}$ in $\rhorun{p}{\mu \land \eta}{\overline{m}^j}$ visits $F_{(\mu \land \eta)}$ infinitely often. Since $F_{(\mu \land \eta)} \subset F_{(\Rarg{\mu}{\eta})}$, it also visits $F_{(\Rarg{\mu}{\eta})}$ infinitely often. We hence use the couples in $\smash{R_{d^j}}$ and their subruns to complete our run on $\overline{m}$.

Any branch $\beta$ in the resulting run eventually reaches a branch in a copied subrun on $\overline{m}^k$ for some $1 \leq k \leq j$ specific to $\beta$. As previously argued, these subbranches visit $F_{(\Rarg{\mu}{\eta})}$ infinitely often. This makes our run on $\overline{m}$ $F_{(\Rarg{\mu}{\eta})}$-accepting.

The proof for the subcase $\Uarg{\mu}{\eta}$ is analogous to the one for $\Rarg{\mu}{\eta}$. It is simply a matter of swapping formulas in some of the symbols used. As an example, equation \eqref{eqnReleaseDecompSuper} maintains its overall form, but becomes:
\begin{equation*}
\transfunc{p}{\Uarg{\mu}{\eta}}{m} =
\bigvee_{d \in D_\eta}
\Big(\bigwedge_{c \in C_d} R_{d, c}\Big)
\lor
\bigvee_{d \in D_\mu}
\Big(\bigwedge_{c \in C_d} R_{d, c} \land (p, \Uarg{\mu}{\eta})\Big)
\text{.}
\end{equation*}

The sets $F_\mu$ and $F_\eta$ are obviously included in the set $F_{(\Uarg{\mu}{\eta})}$.
\showboxcond{} \end{proof}


\begin{lem} \label{lemPredicateToRun}
For any \LTLE{} formula $\psi$ and any partial function $p : V \rightarrow \image{}$ that assigns every free variable in $\psi$, if the predicate $\logicprop{p}{\psi}{\overline{m}}$ is true for a trace $\overline{m}$, then $\rho$ admits an $F_\psi$-accepting run $\rhorun{p}{\psi}{\overline{m}}$  with root $(p, \psi)$.
\end{lem}

\begin{proof}
As with Lemma \ref{lemRunToPredicate}, we proceed by strong induction on the temporal depth of \LTLE{} formulas. Our approach is based on the decompositions of $\transfunc{p}{\psi}{m_1}$ and $\logicprop{p}{\psi}{\overline{m}}$, which are given by Lemmas \ref{lemRhoDecomp} and \ref{lemLogicDecomp} respectively:
\begin{equation} \tag{\ref{eqnRhoDecompLem}}
\transfunc{p}{\psi}{m_1} =
\bigvee_{d \in D}
\Big(
\bigwedge_{e \in E_d} \transfunc{p_{d, e}}{\varepsilon_{d, e}}{m_1}
\land
\bigwedge_{n \in N_d} (p_{d, n}, \psi_{d, n})
\land
\bigwedge_{a \in A_d} (\varnothing, \accstate_{d, a})
\Big)
\text{;}
\end{equation}
$\logicprop{p}{\psi}{\overline{m}} \Leftrightarrow$\\
\begin{equation} \tag{\ref{eqnLogicDecompLem}}
\bigvee_{d \in D}
\Big(
\bigwedge_{e \in E_d} \logicprop{p_{d, e}}{\varepsilon_{d, e}}{\overline{m}}
\land
\bigwedge_{n \in N_d} \logicprop{p_{d, n}}{\psi_{d, n}}{\overline{m}^2}
\land
\bigwedge_{a \in A_d} \logicprop{\varnothing}{\accstate_{d, a}}{\overline{m}^2}
\Big)
\text{.}
\end{equation}

For any $d \in D$, $e \in E_d$,  $n \in N_d$ and $a \in A_d$, the terms $p_{d, e}$, $p_{d, n}$, $\varepsilon_{d, e}$,  $\psi_{d, n}$ and $\accstate_{d, a}$ in \eqref{eqnRhoDecompLem} can be assumed to be identical to their counterpart in \eqref{eqnLogicDecompLem} due to Lemma \ref{lemDecompEquiv}.

\begin{flushleft}
\textit{Base case}: Lemma \ref{lemPredicateToRun} holds for any formula of depth 0.
\end{flushleft}

Suppose $\timedepth{\psi} = 0$. We know from Lemmas \ref{lemRhoDecomp} and \ref{lemLogicDecomp} that the sets $N_d$ and $A_d$ in this special case, are empty for every $d \in D$. Thus, we can rewrite equations \eqref{eqnRhoDecompLem} and \eqref{eqnLogicDecompLem} as follows:
\begin{equation} \tag{\ref{eqnRhoDecompLemBase}}
\transfunc{p}{\psi}{m_1} =
\bigvee_{d \in D}
\Big(
\bigwedge_{e \in E_d} \transfunc{p_{d, e}}{\varepsilon_{d, e}}{m_1}
\Big)
\text{;}
\end{equation}
\begin{equation} \tag{\ref{eqnLogicDecompLemBase}}
\logicprop{p}{\psi}{\overline{m}} \leftrightarrow
\bigvee_{d \in D}
\Big(
\bigwedge_{e \in E_d} \logicprop{p_{d, e}}{\varepsilon_{d, e}}{\overline{m}}
\Big)
\text{.}
\end{equation}

If the predicate $\logicprop{p}{\psi}{\overline{m}}$ holds, then there exists a $d^* \in D$ for which $\logicprop{p_{d^*, e}}{\varepsilon_{d^*, e}}{\overline{m}}$ is true for every $e \in E_{d^*}$. It follows from the definition of $\rho$ that the output of a term $\transfunc{p_{d, e}}{\varepsilon_{d, e}}{m_1}$ in \eqref{eqnRhoDecompLemBase} is $(\varnothing, \accstate)$ for  $d = d^*$. Hence, this  couple satisfies the right-hand side of \eqref{eqnRhoDecompLemBase} and can thus be the only child of the root $(p, \psi)$ in our run. Since the only successor of a node $(\varnothing, \accstate)$ is itself (for any input message), the resulting run has a single  branch that visits the node $(\varnothing, \accstate)$ infinitely often. Since the state $\accstate$ is in $F_\psi$, this run on $\overline{m}$ is $F_\psi$-accepting.

\begin{flushleft}
\textit{Induction step}: For some natural $t > 0$, if Lemma \ref{lemPredicateToRun} holds for any formula of depth less than $t$, then it also holds for any formula of depth $t$.
\end{flushleft}

Suppose $\timedepth{\psi} = t$. The sets $N_d$ and $A_d$ may not be empty for all $d \in D$, so equations \eqref{eqnRhoDecompLem} and \eqref{eqnLogicDecompLem} must be used. If $\logicprop{p}{\psi}{\overline{m}}$ holds, then there exists a $d^* \in D$ for which the predicates $\logicprop{p_{d^*, e}}{\varepsilon_{d^*, e}}{\overline{m}}$, $\logicprop{p_{d^*, n}}{\psi_{d^*, n}}{\overline{m}^2}$  and  $\logicprop{\varnothing}{\accstate_{d, a}}{\overline{m}^2}$ in \eqref{eqnLogicDecompLem} hold for any $e \in E_{d^*}$,  $n \in N_{d^*}$ and $a\in A_{d^*}$.

 The accepting run for $(p, \psi)$ will include every couple indexed by $d^*$, in $E_{d^*}$, $N_{d^*}$ and $A_{d^*}$ .  The base case already shows that any term $\transfunc{p_{d, e}}{\varepsilon_{d, e}}{m_1}$ indexed by $d^*$ in \eqref{eqnRhoDecompLem} outputs the accepting state $(\varnothing, \accstate)$. If $N_{d^*}$ and $A_{d^*}$ are not empty, we must include their content in the children of the root $(p, \psi)$ in the accepting run.

We also include any couple $(p_{d, n}, \psi_{d, n})$ indexed by $d^*$ in \eqref{eqnRhoDecompLem} in order to satisfy the right-hand side of this equation. We know that the branch generated by the child node $(\varnothing, \accstate)$ visits $F_\psi$ infinitely often. It remains to define an $F_\psi$-accepting subrun on the trace $\overline{m}^2$ for every child node $(p_{d^*, n}, \psi_{d^*, n})$. It remains to show that the elements of these sets also generate an accepting run.

We consider first the elements of $N_{d^*}$.

Since every formula $\psi_{d^*, n}$ is a subformula of $\psi$ (Lemma \ref{lemRhoDecomp}),  the inclusion $\smash{F_{\psi_{d^*, n}}} \subseteq F_\psi$ holds. Therefore, for any child node $(p_{d^*, n}, \psi_{d^*, n})$, it suffices to show that its subrun is $\smash{F_{\psi_{d^*, n}}}$-accepting. Lemma \ref{lemRhoDecomp} also tells us that $p_{d^*, n}$ assigns every free variable in $\psi_{d^*, n}$ for any $n \in N_{d^*}$. As such, every predicate $\logicprop{p_{d^*, n}}{\psi_{d^*, n}}{\overline{m}^2}$ in \eqref{eqnLogicDecompLem} satisfies the condition of Lemma \ref{lemPredicateToRun}.

The remainder of the argument is made easy by the induction hypothesis. Indeed, by Lemma \ref{lemRhoDecomp}, a formula $\psi_{d^*, n}$ fits only one of three possible cases:

\begin{flushleft}
1) $\timedepth{\psi_{d^*, n}} < t$
\end{flushleft}

The induction hypothesis directly applies to $\logicprop{p_{d^*, n}}{\psi_{d^*, n}}{\overline{m}^2}$. Thus, $\rho$ admits an $\smash{F_{\psi_{d^*, n}}}$-accepting run on $\overline{m}^2$ and with root $(p_{d^*, n}, \psi_{d^*, n})$ as desired. As was the case for lemma \ref{lemRunToPredicate}, this case also implies the validity of any formula in $A_{d^*}$.

\begin{flushleft}
2) $\timedepth{\psi_{d^*, n}} = t$ and $\psi_{d^*, n} = \Uarg{\mu}{\eta}$ where $\mu, \eta \in \subformulaset{\psi}$
\end{flushleft}

We simply follow the proof of Lemma \ref{lemRunToPredicate} for this case, but in reverse. First, if $\logicprop{p_{d^*, n}}{\Uarg{\mu}{\eta}}{\overline{m}^2}$ holds, then by the definition of $\U$, for some $j \geq 2$, $\logicprop{p_{d^*, n}}{\eta}{\overline{m}^j}$ holds and for every $2 \leq i < j$, $\logicprop{p_{d^*, n}}{\mu}{\overline{m}^i}$ also holds. Next, because the temporal depths of $\mu$ and $\eta$ are less than $t$, the induction hypothesis applies. Hence, $\rho$ admits an $F_\eta$-accepting run on $\overline{m}^j$ and $F_\mu$-accepting runs on $\overline{m}^i$ for every $2 \leq i < j$. Finally, by Lemma \ref{lemSuperrunUR}, $\rho$ admits an $F_{(\Uarg{\mu}{\eta})}$-accepting run on $\overline{m}^2$ as desired.

\begin{flushleft}
3) $\timedepth{\psi_{d^*, n}} = t$ and $\psi_{d^*, n} = \Rarg{\mu}{\eta}$ where $\mu, \eta \in \subformulaset{\psi}$
\end{flushleft}

As with(2), we follow the proof of Lemma \ref{lemRunToPredicate} for the current case in reverse. Note though that if $\logicprop{p_{d^*, n}}{\Rarg{\mu}{\eta}}{\overline{m}^2}$ holds, two possibilities arise:
\begin{itemize}
\item[$\bullet$] $\forall \ i \geq 2$, $\logicprop{p_{d^*, n}}{\eta}{\overline{m}^i}$ is true;
\item[$\bullet$] $\exists \ j \geq 2$ for which $\logicprop{p_{d^*, n}}{\mu \land \eta}{\overline{m}^j}$ is true, and $\logicprop{p_{d^*, n}}{\eta}{\overline{m}^i}$ is true $\forall \ 2 \leq i < j$.
\end{itemize}

Since the depth of $\mu \land \eta$,  is less than $t$, the induction hypothesis applies for both possibilities. 
\showboxcond{}
\end{proof}

\begin{prs} \label{prpPredicateToRun}
For any \LTLE{} formula $\varphi$ devoid of free variables, and let  $\overline{m}$ be a message trace, if $\overline{m}$ satisfies $\varphi$, then the automaton $A_\varphi$ accepts $\overline{m}$.
\end{prs}

\begin{proof}
All variables in $\varphi$ are bound by quantifiers ($\exists$ or $\forall$), so the statement ``$\overline{m}$ satisfies $\varphi$'' is represented by $\logicprop{\varnothing}{\varphi}{\overline{m}}$. It follows from Lemma \ref{lemPredicateToRun} that there exists an $F_\varphi$-accepting run of $\rho$ on $\overline{m}$ with root $(\varnothing, \varphi)$, which, by definition, is an accepting run of $A_\varphi$ on $\overline{m}$.
\showboxcond{}\end{proof}

\begin{thm}[adapted from Vardi]
Given any \LTLE{} formula $\varphi$, one can build a modified alternating B\"uchi automaton $A_\varphi = (\Sigma, V, \image{}, S, s^0, \rho, F)$, where $\Sigma = M$ and $|S|$ is in $O(|\varphi|)$, such that the language recognized by $A_\varphi$ is exactly the set of message traces satisfying the formula $\varphi$.
\end{thm}

\begin{proof}
Immediate from the definition of $A_\varphi$ and Propositions \ref{prpRunToPredicate} and \ref{prpPredicateToRun}.
\showboxcond{}\end{proof}


\section{Conclusion and Future Works}\label{sect:conclu}
In this paper, we propose a new type of finite alternating automata which recognizes \LTLE{} formul\ae{} an show the process of constructing such an automaton from an \LTLE{} formula.  Our automaton allows for formul\ae{} in the highly expressive logic, \LTLE{} formal logics to be easily stated in a concise and easy to understand formalism. We are currently developing and implementing a verification algorithm that will allow this new tool to be put to practical use.

\bibliographystyle{splncs03}
\bibliography{paper}
\end{document}